\documentclass[11pt]{article}

\usepackage[a4paper,margin=1in]{geometry}
\usepackage{amsthm}
\usepackage{amsmath, amssymb}

\usepackage[pdfpagelabels,bookmarks=false]{hyperref}

\hypersetup{colorlinks, linkcolor=darkblue, citecolor=darkgreen, urlcolor=darkblue}
\newcommand{\footremember}[2]{%
    \footnote{#2}
    \newcounter{#1}
    \setcounter{#1}{\value{footnote}}%
}
\newcommand{\footrecall}[1]{%
    \footnotemark[\value{#1}]%
}

\usepackage{times}

\usepackage{anyfontsize}
\usepackage{complexity}

\usepackage{xspace}

\usepackage[inline]{enumitem}
\setlist[enumerate,1]{label=(\roman*), leftmargin=2.2em}
\setlist[enumerate,2]{label=(\alph*)}
\setlist{nosep,topsep=0.1em}
\setlist[itemize,1]{label={\bfseries--}}

\usepackage{mathtools}

\makeatletter
\newtheorem*{rep@theorem}{\rep@title}
\newcommand{\newreptheorem}[2]{%
	\newenvironment{rep#1}[1]{%
		\def\rep@title{\cref{##1}}%
		\begin{rep@theorem}}%
		{\end{rep@theorem}}}
\makeatother

\newreptheorem{theorem}{Theorem}
\newreptheorem{lemma}{Lemma}
\newreptheorem{corollary}{Corollary}

\usepackage{mdframed}

\makeatletter
\let\@@pmod\pmod
\DeclareRobustCommand{\pmod}{\@ifstar\@pmods\@@pmod}
\def\@pmods#1{\mkern8mu({\operator@font mod}\mkern 6mu#1)}
\makeatother

\makeatletter
\let\@@mod\mod
\DeclareRobustCommand{\mod}{\@ifstar\@mods\@@mod}
\def\@mods#1{\mkern8mu{\operator@font mod}\mkern 6mu#1}
\makeatother

\usepackage[color=darkgreen!40!white]{todonotes}
\setuptodonotes{size=\scriptsize}

\usepackage{xcolor}
\definecolor{darkblue}{rgb}{0,0,0.38}
\definecolor{darkred}{rgb}{0.6,0,0}
\definecolor{darkgreen}{rgb}{0.1,0.35,0}

\usepackage[
backend=biber,
style=alphabetic,
citestyle=alphabetic,
maxalphanames=6,
maxcitenames=99,
mincitenames=98,
maxbibnames=99,
giveninits=true,
url=false,
doi=true,
isbn=true,
backref=true
]{biblatex}

\usepackage{xpatch}

\makeatletter
\patchcmd\blx@bblinput{\blx@blxinit}
                      {\blx@blxinit
                      }{}{\fail}
\makeatother
\makeatletter
\DeclareFieldFormat{eprint:arxiv}{%
  arXiv\addcolon\space
  \ifhyperref
    {\href{https://arxiv.org/abs/#1}{%
       \nolinkurl{#1}%
       \iffieldundef{eprintclass}
     {}
     {\addspace\mkbibbrackets{\thefield{eprintclass}}}}}
    {\nolinkurl{#1}
     \iffieldundef{eprintclass}
       {}
       {\addspace\mkbibbrackets{\thefield{eprintclass}}}}}
\makeatother

\usepackage[vlined,ruled,algo2e,linesnumbered]{algorithm2e}
\SetAlCapHSkip{0.2em}
\SetAlgoInsideSkip{smallskip}
\SetCustomAlgoRuledWidth{0.95\textwidth}
\makeatletter
\patchcmd{\@algocf@start}{%
  \begin{lrbox}{\algocf@algobox}%
}{%
  \rule{0.025\textwidth}{\z@}%
  \begin{lrbox}{\algocf@algobox}%
  \begin{minipage}{0.95\textwidth}%
}{}{}
\patchcmd{\@algocf@finish}{%
  \end{lrbox}%
}{%
  \end{minipage}%
  \end{lrbox}%
}{}{}
\makeatother

\usepackage{xfrac}
\ExplSyntaxOn
\DeclareRestrictedTemplate { xfrac } { text } { math }
  {
    numerator-font      = \number \fam ,
    slash-symbol        = /            ,
    slash-symbol-font   = \number \fam ,
    denominator-font    = \number \fam ,
    scale-factor        = 0.7          ,
    scale-relative      = false        ,
    scaling             = true         ,
    denominator-bot-sep = 0 pt         ,
    math-mode           = true         ,
    phantom             = ( %
  }
\DeclareInstance { xfrac } { mathdefault } { math }
  { numerator-top-sep = 0pt }
\ExplSyntaxOff

\usepackage{wrapfig}

\usepackage[margin=10pt,font=small,labelfont=bf]{caption}
\usepackage[margin=10pt,font=small,labelfont=bf]{subcaption}

\usepackage{graphicx}

\usepackage{tikz}
\usetikzlibrary{calc}

\tikzstyle{patternSquare} = [draw, thick, rectangle, minimum width=.9cm, minimum height=.9cm, align=center, inner sep=.5mm, text width=.8cm, execute at begin node=, font=\small, text height=1.5ex, text depth=0.25ex]

\newcommand{\ksum}[1][]{\mathbin{\oplus_{#1}}}
\newcommand{\pivot[1]}[]{\operatorname{pivot}_{#1}}

\makeatletter
\DeclareRobustCommand{\cev}[1]{%
  {\mathpalette\do@cev{#1}}%
}
\newcommand{\do@cev}[2]{%
  \vbox{\offinterlineskip
    \sbox\z@{$\m@th#1 x$}%
    \ialign{##\cr
      \hidewidth\reflectbox{$\m@th#1\vec{}\mkern4mu$}\hidewidth\cr
      \noalign{\kern-\ht\z@}
      $\m@th#1#2$\cr
    }%
  }%
}
\makeatother

\makeatletter
 \newcommand{\labeltarget}[1]{\Hy@raisedlink{\hypertarget{#1}{}}}
\makeatother
\makeatletter
 \newcommand{\linkdest}[1]{\Hy@raisedlink{\hypertarget{#1}{}}}
\makeatother
\newcommand{\MCCTU}{\hyperlink{prb:MCCTU}{MCCTU}\xspace}
\newcommand{\GCTUF}{\hyperlink{prb:GCTUF}{GCTUF}\xspace}
\newcommand{\RGCTUF}[1][R]{\hyperlink{prb:R-GCTUF}{$#1$-GCTUF}\xspace}
\newcommand{\GCLF}{\hyperlink{prb:GCLF}{GCLF}\xspace}
\newcommand{\GCTU}{\hyperlink{prb:GCTU}{GCTU}\xspace}
\newcommand{\GCC}{\hyperlink{prb:GCC}{GCC}\xspace}
\newcommand{\XLC}{\hyperlink{prb:XLC}{XLC}\xspace}
\newcommand{\GCLO}{\hyperlink{prb:GCLO}{GCLO}\xspace}

\newcommand{\ab}{(\alpha,\beta)}
\newcommand{\abprime}{(\alpha',\beta')}

\newcommand{\Z}{\mathbb{Z}}

\usepackage[capitalize, nameinlink]{cleveref}

\newtheorem*{theorem*}{Theorem}
\newtheorem{theorem}{Theorem}
\newtheorem{lemma}[theorem]{Lemma}
\newtheorem*{lemma*}{Lemma}
\newtheorem{conjecture}[theorem]{Conjecture}

\newtheorem{proposition}[theorem]{Proposition}
\newtheorem{definition}[theorem]{Definition}
\newtheorem{remark}[theorem]{Remark}
\newtheorem{corollary}[theorem]{Corollary}

\newtheorem{observation}[theorem]{Observation}

\crefname{theorem}{Theorem}{Theorems}
\crefname{conjecture}{Conjecture}{Conjectures}
\Crefname{lemma}{Lemma}{Lemmas}
\Crefname{claim}{Claim}{Claims}
\Crefname{fact}{Fact}{Facts}
\Crefname{remark}{Remark}{Remarks}
\Crefname{observation}{Observation}{Observations}
\Crefname{line}{Line}{Lines}
\Crefname{figure}{Figure}{Figures}
\crefalias{AlgoLine}{line}

\usepackage{multicol}

\newcommand{\Pinar}{\widehat{\Pi}}
\newcommand{\quot}[3][\Big]{{{}^{\textstyle #2}\!#1/\!_{\textstyle #3}}}
\newcommand{\inlinequot}[2]{\sfrac{#1}{#2}}

\title{Advances on Strictly $\Delta$-Modular IPs%
\thanks{%
Funded through the Swiss National Science Foundation grants 200021\_184622 and P500PT\_206742, the European Research Council (ERC) under the European Union's Horizon 2020 research and innovation programme (grant agreement No 817750), and the Deutsche Forschungsgemeinschaft (DFG, German Research Foundation) under Germany's Excellence Strategy~--~EXZ-2047/1~--~390685813.%
}}

\author{
Martin N{\"a}gele%
\thanks{%
Research Institute for Discrete Mathematics and Hausdorff Center for Mathematics, University of Bonn, Bonn, Germany.
Email: \href{mailto:mnaegele@uni-bonn.de}%
{mnaegele@uni-bonn.de}.%
}%
\and
Christian N{\"o}bel%
\footremember{ETH}{
Department of Mathematics, ETH Zurich, Zurich, Switzerland.
Email: $\{$\href{mailto:cnoebel@ethz.ch}{cnoebel}, \href{mailto:rtorres@ethz.ch}{rtorres}, \href{mailto:ricoz@ethz.ch}{ricoz}$\}$@ethz.ch.%
}
\and
Richard Santiago%
\footrecall{ETH}%
\and
Rico Zenklusen
\footrecall{ETH}%
}

\date{}

\begin{document}

\maketitle

\thispagestyle{empty}
\addtocounter{page}{-1}

\begin{abstract}
There has been significant work recently on integer programs (IPs) $\min\{c^\top x \colon Ax\leq b,\,x\in \mathbb{Z}^n\}$ with a constraint marix $A$ with bounded subdeterminants.
This is motivated by a well-known conjecture claiming that, for any constant $\Delta\in \mathbb{Z}_{>0}$, $\Delta$-modular IPs are efficiently solvable, which are IPs where the constraint matrix $A\in \mathbb{Z}^{m\times n}$ has full column rank and all $n\times n$ minors of $A$ are within $\{-\Delta, \dots, \Delta\}$.
Previous progress on this question, in particular for $\Delta=2$, relies on algorithms that solve an important special case, namely \emph{strictly $\Delta$-modular IPs}, which further restrict the $n\times n$ minors of $A$ to be within $\{-\Delta, 0, \Delta\}$.
Even for $\Delta=2$, such problems include well-known combinatorial optimization problems like the minimum odd/even cut problem.
The conjecture remains open even for strictly $\Delta$-modular IPs. Prior advances were restricted to prime $\Delta$, which allows for employing strong number-theoretic results.

In this work, we make first progress beyond the prime case by presenting techniques not relying on such strong number-theoretic prime results.
In particular, our approach implies that there is a randomized algorithm to check feasibility of strictly $\Delta$-modular IPs in strongly polynomial time if $\Delta\leq4$.

\end{abstract}

\begin{tikzpicture}[overlay, remember picture, shift = {(current page.south east)}]
\coordinate (anchor) at (0,0);
\node[anchor=south east, outer sep=5mm] at (anchor) {
\begin{tikzpicture}[outer sep=0] %
\node (ERC) {\includegraphics[height=13mm]{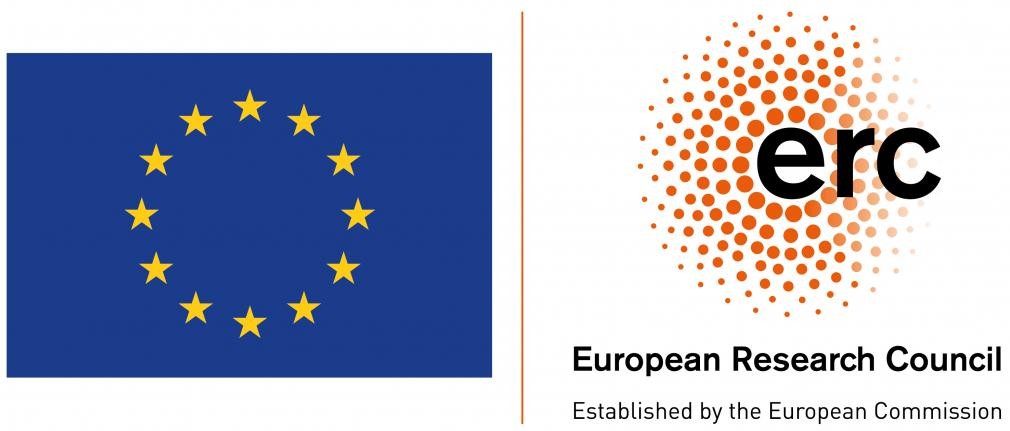}};
\node[left=5mm of ERC] (SNSF) {\includegraphics[height=7mm]{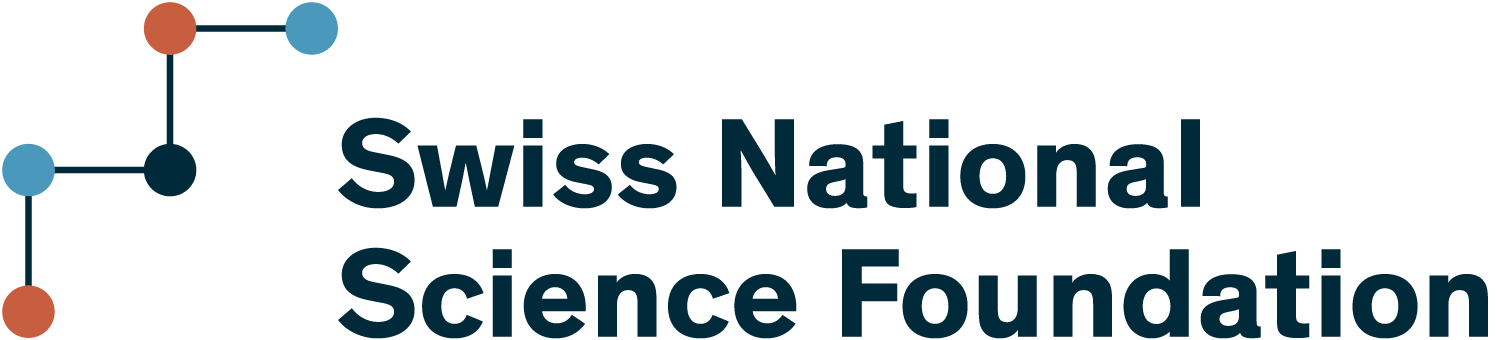}};
\node[right=5mm of ERC] (DFG) {\includegraphics[height=5mm]{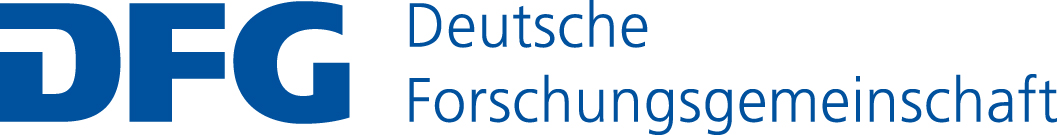}};%
\end{tikzpicture}
};
\end{tikzpicture}

\newpage

\section{Introduction}

Integer Programs (IPs) $\min\{c^\top x\colon Ax\leq b,\, x\in\mathbb{Z}^n\}$ are a central \NP-hard problem class in Combinatorial Optimization.
There is substantial prior work and interest in identifying special classes of polynomial-time solvable IPs while remaining as general as possible.
One of the best-known such classes are IPs with a constraint matrix that is \emph{totally unimodular} (TU), i.e., the determinant of any of its square submatrices is within $\{-1,0,1\}$.
A long-standing open conjecture in the field is whether this result can be generalized to $\Delta$-modular constraint matrices for constant $\Delta$.
Here, we say that a matrix $A\in\mathbb{Z}^{k\times n}$ is \emph{$\Delta$-modular} if it has full column rank and all $n\times n$ submatrices have determinants in $\{-\Delta,\ldots,\Delta\}$.%
\footnote{A weaker variant of the conjecture claims efficient solvability of IPs with \emph{totally $\Delta$-modular} constraint matrices, where all subdeterminants are bounded by $\Delta$ in absolute value.
The conjecture involving $\Delta$-modular matrices implies the weaker variant.
Indeed, an IP $\min\{c^\top x \colon Ax \leq b, x\in \mathbb{Z}^n\}$ with a totally $\Delta$-modular constraint matrix can be reformulated as $\min\{c^\top (x^+-x^-) \colon A(x^+-x^-)\leq b, x^+, x^- \in \mathbb{Z}_{\geq 0}^n\}$.
It is not hard to see that the constraint matrix of the new LP remains totally $\Delta$-modular; moreover, it has full column rank because of the non-negativity constraints.
}
For brevity, we call an IP with $\Delta$-modular constraint matrix a \emph{$\Delta$-modular IP}.
We recap the above-mentioned conjecture below.
Unfortunately, we do not know its precise origin; it may be considered folklore in the field.
\begin{conjecture}\label{conj:deltaMod}
For constant $\Delta\in \mathbb{Z}_{\geq 0}$, $\Delta$-modular IPs can be solved in polynomial time.
\end{conjecture}

First progress on \cref{conj:deltaMod} was made by \textcite{artmann_2017_strongly}, who showed that it holds for $\Delta=2$ (the bimodular case).
\textcite{fiorini_2021_integerPrograms} show that the conjecture is true for an arbitrary constant $\Delta$ under the extra condition that the constraint matrix has at most two non-zero entries per row or column.
Through a non-trivial extension of the techniques in~\cite{artmann_2017_strongly}, it was shown by \textcite{nagele_2022_congruency} that there is a randomized algorithm to check feasibility of an IP with a strictly 3-modular constraint matrix in polynomial time.
Here, a matrix $A\in\Z^{k\times n}$ is called \emph{strictly $\Delta$-modular} if it has full column rank and all its $n\times n$ submatrices have determinants in $\{-\Delta, 0, \Delta\}$.

As a key ingredient, all these prior approaches solve certain combinatorial optimization problems with congruency constraints.
This is not surprising, as even strictly $\Delta$-modular IPs include the following class of \emph{MCCTU problems}:%
\footnote{%
To capture an \MCCTU{} problem as a strictly $\Delta$-modular IP, replace each congruency constraint $\gamma_i^\top x \equiv r_i\pmod*{m_i}$ by an equality constraint $\gamma_i^\top x + m_i y_i = r$ with $y_i\in\mathbb{Z}$.
The corresponding constraint matrix is strictly $\Delta$-modular for $\Delta=\prod_{i=1}^q m_i$.%
}
\begin{mdframed}[innerleftmargin=0.5em, innertopmargin=0.5em, innerrightmargin=0.5em, innerbottommargin=0.5em, userdefinedwidth=0.95\linewidth, align=center]
	{\textbf{Multi-Congruency-Constrained TU Problem (\MCCTU{}\linkdest{prb:MCCTU}):}}
	\sloppy Let $T\in\mathbb{Z}^{k\times n}$ be TU, $b\in\mathbb{Z}^k$, $c\in \mathbb{R}^k$, $m\in \mathbb{Z}^q_{>0}$, $\gamma_i\in \mathbb{Z}^n$ for $i\in [q]$, $r\in \mathbb{Z}^q$.
Solve
$$
	\min\{c^\top x \colon Tx \leq b,\, \gamma_i^\top x \equiv r_i \pmod*{m_i} \;\forall i\in [q],\, x\in\mathbb{Z}^n\}\enspace.
$$
\end{mdframed}
Unless mentioned otherwise, we assume that in the context of \MCCTU{} problems, $q$ and $m_i$ are constant.
Even \MCCTU with just a single congruency constraint, i.e., $q=1$, already contains the classical and well-studied odd and even cut problems, and, more generally, the problem of finding a minimum cut whose number of vertices is $r \pmod*{m}$.
(See~\cite{padberg_1982_odd,barahona_1987_construction,groetschel_1984_corrigendum,goemans_1995_minimizing,nagele_2018_submodular,nagele_2020_newContraction} for related work.)
It can also capture the minimum $T$-join problem, congruency-constrained flow problems, and many other problems linked to TU matrices.

Combinatorial optimization problems with congruency constraints are highly non-trivial and many open questions remain.
As they are already captured by strictly $\Delta$-modular IPs, this motivates the following weakening of \cref{conj:deltaMod}.

\begin{conjecture}\label{conj:strictlyDeltaMod}
Strictly $\Delta$-modular IPs can be solved in polynomial time for constant $\Delta\in \mathbb{Z}_{\geq 0}$.
\end{conjecture}

Even resolving this weaker  conjecture would settle several open problems, including congruency-constrained min cuts (in both directed and undirected graphs), or the problem of efficiently and deterministically finding a perfect matching in a red/blue edge-colored bipartite graph such that the number of red matching edges is $r \pmod*{m}$.
(This is a simplified version of the famous red-blue matching problem, where the task is to find a perfect matching with a specified number of red edges; for both versions, randomized algorithms are known.)
Interestingly, for the bimodular case ($\Delta=2$), a result by \textcite{veselov_2009_integer} implies that \cref{conj:deltaMod} and \cref{conj:strictlyDeltaMod} are equivalent (see~\cite{artmann_2017_strongly}).

Our goal is to shed further light on \cref{conj:strictlyDeltaMod} and overcome some important hurdles of prior approaches.
In a first step, we note that a positive resolution of \cref{conj:strictlyDeltaMod} does not only imply efficient solvability of \MCCTU problems, but also vice versa, and this reduction works in strongly polynomial time.
\begin{lemma}\label{lem:redToMCCTU}
Let $\Delta>0$.
Every strictly $\Delta$-modular IP can, in strongly polynomial time, be reduced to an \MCCTU problem with moduli $m_i$ such that $\Delta=\prod_{i=1}^q m_i$.
\end{lemma}

Without the strongly polynomial time condition, this also follows from very recent work of~\textcite[Lemma 4]{gribanov_2022_delta-modular}.

Further, we are interested in making progress regarding the feasibility version of \cref{conj:strictlyDeltaMod}, i.e., efficiently deciding whether a strictly $\Delta$-modular IP is feasible.
Prior approaches settle this question for $\Delta=2$~\cite{artmann_2017_strongly} and---using a randomized algorithm---%
for $\Delta=3$~\cite{nagele_2022_congruency}.
A main hurdle to extend these is that they crucially rely on $\Delta$ being prime, for example through the use of the Cauchy-Davenport Theorem.
Our main contribution here is to address this.
In particular, we can check feasibility for $\Delta=4$ with a randomized algorithm, which is the first result in this context for non-prime $\Delta$.
More importantly, our techniques will hopefully prove useful for future advances on this challenging question.
\begin{theorem}\label{thm:strictly-delta}
There exists a strongly polynomial-time randomized algorithm to find a feasible solution of a strictly $4$-modular IP, or detect that it is infeasible.
\end{theorem}

We remark that the randomization appearing in the above theorem comes from the fact that one building block of our result is a reduction to a problem class that includes the aforementioned congruency-constrained red/blue-perfect matching problem, for which only randomized approaches are known.

\subsection{Group-constrained problems and proof strategy for \cref{thm:strictly-delta}}

To show \cref{thm:strictly-delta}, we exploit its close connection to \MCCTU.
Capturing the congruency constraints of an \MCCTU problem through an abelian group constraint, we attain the following \emph{group-constrained TU feasibility problem}.
\begin{mdframed}[innerleftmargin=0.5em, innertopmargin=0.5em, innerrightmargin=0.5em, innerbottommargin=0.5em, userdefinedwidth=0.95\linewidth, align=center]
	{\textbf{Group-Constrained TU Feasibility (\GCTUF{}\linkdest{prb:GCTUF}):}}
	\sloppy Let $T\in\mathbb{R}^{k\times n}$ be a TU matrix, let $b\in\mathbb{Z}^k$, let $(G,+)$ be a finite abelian group, and let $\gamma\in G^{n}$ and $r\in G$. The task is to show infeasibility or find a solution of the system
	$$
	Tx \leq b,\ \gamma^\top x = r,\ x\in\mathbb{Z}^n \enspace.
	$$
\end{mdframed}
Here, the scalar product $\gamma^\top x$ denotes the linear combination of the group elements $\gamma_1,\ldots,\gamma_n$ with multiplicities $x_1,\ldots,x_n$ in $G$.
Group constraints generalize congruency constraints, which are obtained in the special case where $G$ is cyclic.
More generally, by the fundamental theorem of finite abelian groups, a finite abelian group $G$ is, up to isomorphism, a direct product of cyclic groups.
Hence, a group constraint can be interpreted as a set of congruency constraints and vice versa.
Thus, \GCTUF and \MCCTU feasibility are two views on the same problem.
We stick to \GCTUF mostly for convenience of notation. Moreover, the \GCTUF setting also allows for an elegant use of group-related results later on.
One may assume that the group is given through its multiplication table (the \emph{Cayley table}).
In fact, the precise group representation is not of great importance to us.
Concretely, for constant $\Delta$, strictly $\Delta$-modular IP feasibility problems reduce to \GCTUF problems with a  constant size group.
Many of our polynomial-time algorithmic results can even be extended to settings where the group size is not part of the input, and access to group operations is provided through an oracle.

By a slight extension of \cref{lem:redToMCCTU} (see \cref{sec:reductionToGCTU}) and the aforementioned equivalent viewpoint of multiple congruency constraints and a group constraint, in order to prove \cref{thm:strictly-delta}, it is enough for us to show the equivalent statement below.
\begin{theorem}\label{thm:GCTUFmod4}
	There exists a strongly polynomial time randomized algorithm for \GCTUF{} problems with a group of cardinality at most $4$.
\end{theorem}

On a high level, we follow a well-known strategy for TU-related problems by employing Seymour's decomposition~\cite{seymour_1980} to decompose the problem into problems on simpler, more structured TU matrices.
(See, e.g.,~\cite{dinitz_2014_matroid,artmann_2017_strongly,aprile_2021_regular,nagele_2022_congruency}.)
Roughly speaking, Seymour's decomposition states that a TU matrix is either very structured---in which case we call it a \emph{base block}---or can be decomposed into smaller TU matrices through a small set of well-defined operations.
(See the discussion following \cref{thm:TUdecomp} for more details.)
The use of Seymour's decomposition typically comes with two main challenges, namely
\begin{enumerate*}
\item solving the base block cases, and
\item propagating solutions of the base block cases back through the decomposition efficiently to solve the original problem.
\end{enumerate*}
First, we show that this propagation can be done efficiently for our problem.
\begin{theorem}\label{thm:reductionToBaseBlocksG4}
	Let $G$ be an abelian group of size at most $4$.
	Given an oracle for solving
	{%
		base block \GCTUF{} problems with group $G$, %
	}%
	we can solve
	{%
		\GCTUF{} problems with group $G$ %
	}%
	in strongly polynomial time with strongly polynomially many calls to the oracle.
\end{theorem}

In fact, our approach underlying \cref{thm:GCTUFmod4,thm:reductionToBaseBlocksG4} operates in a hierarchy of \GCTUF{} problems with increasingly relaxed group constraints of the form $\gamma^\top x\in R$ for subsets $R\subseteq G$ of increasing size, and allows for proving the above results for such relaxed \GCTUF{} problems for arbitrary constant-size groups $G$ as long as $|G|-|R|\leq 3$.
(See \cref{sec:technicalOverview} for more details.)
In principle, this is along the lines of the approach to congruency-constrained TU problems in~\cite{nagele_2022_congruency}, but incorporates the new viewpoint of group constraints, and additionally improves over earlier results in two ways:
First, our approach applies to arbitrary finite abelian groups, while previous setups heavily relied on the group cardinality being a prime.
Secondly, in the setting with relaxed group constraints, we extend the admissible range of $|G|-|R|$ by one, thus proceeding further in the hierarchy of \GCTUF problems, and newly covering \GCTUF problems with groups of cardinality $4$.

Besides being a key part of our approach, \cref{thm:reductionToBaseBlocksG4} underlines that base block \GCTUF problems are not merely special cases, but play a key role in progress on general \GCTUF problems.
There are only two non-trivial types of such base block \GCTUF problems, namely when the constraint matrix is a so-called \emph{network matrix} or a transpose thereof.
Both cases cover combinatorial problems that are interesting on their own, and their complexity status remains open to date.
If the constraint matrix is a network matrix, \GCTUF can be cast as a circulation problem with a group constraint.
By reducing to and exploiting results of \textcite{camerini_1992_rpp} on exact perfect matching problems, a randomized algorithm for the congruency-constrained case has been presented in~\cite{nagele_2022_congruency}.
We observe that these results extend to the group-constrained setting.
The other base block case, where the constraint matrix is the transpose of a network matrix, can be cast as a group-constrained directed minimum cut problem by leveraging a result in~\cite{nagele_2022_congruency}.
Prior work combined this reduction with results on congruency-constrained submodular minimization~\cite{nagele_2018_submodular} to solve the optimization version of the problem for congruency-constraints of prime power modulus.
We show that the feasibility question on this base block can be solved efficiently on any finite abelian group of constant order, thus circumventing the prime power restriction that is intrinsic in prior approaches.
\begin{theorem}\label{thm:transpose-bb}
	Let $G$ be a finite abelian group. There is a strongly polynomial time algorithm for solving \GCTUF{} problems with group $G$ where the constraint matrix is the transpose of a network matrix.
\end{theorem}

\subsection{Further related work}
The parameter $\Delta$ has been studied from various viewpoints.
While efficient recognition of (totally) $\Delta$-modular matrices is open for any $\Delta \geq 2$, approaches to approximate the largest subdeterminant in absolute value were studied~\cite{summa_2015_largest,nikolov_2015_randomized}.
Also, focusing on more restricted subdeterminant patterns proved useful~\cite{veselov_2009_integer,artmann_2016_nondegenerate,glanzer_2021_abcRecognition}.
Aiming at generalizing a bound of \textcite{heller_1957_linear} for $\Delta=1$, bounds on the maximum number of rows of a $\Delta$-modular matrix were obtained~\cite{glanzer_2018_number,lee_2021_polynomial,averkov_2022_maximal}.
Also, the influence of the parameter $\Delta$ on structure and properties of IPs and polyhedra is multi-faceted (see, e.g.,~\cite{bonifas_2012_subdetDiameter,%
eisenbrand_2017_geometric,%
gribanov2016integer,%
gribanov_2021_lattice,%
gribanov_2021_fptas,%
lee_2020_improvingProximity,%
paat_2021_integralitynumber,%
tardos_1986_strongly} and references therein).

\subsection{Structure of the paper}

We prove the strongly polynomial time reduction from \cref{lem:redToMCCTU} in \cref{sec:reductionToGCTU}.
In \cref{sec:transposedNetworkBaseBlock}, we prove \cref{thm:transpose-bb}.
\cref{sec:technicalOverview} illustrates our approach and new contributions towards \cref{thm:reductionToBaseBlocksG4} on a more technical level, and explains the main new ingredients of our proof.
Throughout \cref{sec:technicalOverview}, we build on several results from~\cite{nagele_2022_congruency} that are proved there for congruency-constrained TU problems, i.e., the case of a cyclic group constraint.
In \cref{sec:BBreduction}, we show how these proofs can be adapted to the group setting.

\section{Reducing to group-constrained problems}\label{sec:reductionToGCTU}

We prove the following slightly strengthened version of \cref{lem:redToMCCTU}.

\begin{lemma}\label{lem:redToMCCTU_general}
Let $\Delta>0$.
Given a strictly $\Delta$-modular IP of the form
$ \min\{c^\top x\colon Ax\leq b, x\in\mathbb{Z}^n\}$, %
one can, in strongly polynomial time, determine an \MCCTU problem
\begin{equation*}\label{eq:MC_problem}
\min\{\bar c^\top y \colon Ty \leq b,\, \gamma_i^\top y \equiv r_i \pmod*{m_i} \;\forall i\in [q],\, y\in\mathbb{Z}^n\}
\end{equation*}
together with a non-singular $n\times n$ submatrix $H$ of $A$ such that the following holds:
\begin{enumerate}
\item\label{lemitem:transf_mod} $\Delta=\prod_{i\in[q]}m_i$.
\item\label{lemitem:transf_obj} $\bar{c}^\top=c^\top H^{-1}$.
\item\label{lemitem:transf_sol} The map $x\mapsto Hx$ is a bijection between feasible solutions of the strictly $\Delta$-modular IP and the \MCCTU problem.
\end{enumerate}
\end{lemma}

We remark that the one-to-one correspondence of feasible solutions given in \cref{lemitem:transf_sol} of \cref{lem:redToMCCTU_general} is (besides \cref{lemitem:transf_mod}) precisely what we need to deduce our main result, \cref{thm:strictly-delta}, from \cref{thm:GCTUFmod4}.
Moreover, \cref{lemitem:transf_obj,lemitem:transf_sol} of \cref{lem:redToMCCTU_general} together imply that $x\mapsto Hx$ is not only a bijection between feasible solutions, but also a bijection between optimal solutions of the two involved problems, so \cref{lem:redToMCCTU} is indeed also implied by \cref{lem:redToMCCTU_general}.

\begin{proof}[Proof of \cref{lem:redToMCCTU_general}]
We show how to transform the given strictly $\Delta$-modular problem into an \MCCTU problem.
Let $H$ be an $n\times n$ submatrix of $A$ with $\lvert\det(H)\rvert=\Delta$.
After a variable transformation to $y=Hx$, we can equivalently rewrite the original integer program in the form
$$
\min \{\bar{c}^\top y\colon Ty \leq b,\, H^{-1}y\in\mathbb{Z}^n\}\enspace,
$$
where $\bar{c}^\top=c^\top H^{-1}$, and $T=AH^{-1}$.
By definition, the map $x\mapsto Hx$ is a bijection between feasible solutions of the original IP and the above problem.
Note that $T$ is unimodular and contains an identity submatrix; hence $T$ is totally unimodular.
To complete the proof, it thus suffices to show that the constraint $H^{-1}y\in\mathbb{Z}^n$ can be transformed to multiple congruency constraints with moduli whose product equals $\Delta$.

To this end, we first write $H^{-1} = H_I + H_F$ with an integer matrix $H_I\in\mathbb{Z}^{n\times n}$ and a fractional matrix $H_F\in [0,1)^{n\times n}$, i.e., matrices whose entries are given by
\[
	(H_I)_{i,j} \coloneqq \left\lfloor \left(H^{-1}\right)_{i,j} \right\rfloor
    \quad\text{and}\quad
    (H_F)_{i,j} \coloneqq \left(H^{-1}\right)_{i,j} - (H_I)_{i,j}\enspace,
\]
where, for $x\in\mathbb{R}$, $\lfloor x\rfloor$ is the integer part of $x$, i.e., the unique number $n\in\mathbb{Z}$ with $n\leq x<n+1$.
Using this decomposition, we obtain
\[
H^{-1}y\in\mathbb{Z}^n \quad\iff\quad H_Iy + H_Fy\in\mathbb{Z}^n \quad\iff\quad H_Fy\in\mathbb{Z}^n\enspace.
\]
Because $\det(H)=\Delta$, we have $\Delta H^{-1}\in\mathbb{Z}^{n\times n}$ by Cramer's rule, and thus also $\widetilde{H}_F\coloneqq \Delta H_F\in\mathbb{Z}^{n\times n}$.
Furthermore, the entries of $\widetilde{H}_F$ are bounded by the constant $\Delta$ in absolute value.
Consequently, using a weakly polynomial time algorithm for computing the Smith normal form of an integer matrix~\cite{kannan_1979_polynomial}, we can in strongly polynomial time determine the Smith normal form of $\widetilde{H}_F$, i.e., we can in strongly polynomial time find  unimodular matrices $S, U\in\mathbb{Z}^{n\times n}$ and integers $\widetilde m_i\in\mathbb{Z}$ such that $D=\operatorname{diag}(\widetilde m_1, \ldots, \widetilde m_n)=S^{-1}\widetilde{H}_F U^{-1}$.
Using this decomposition, we get
\[
	H_Fy\in\mathbb{Z}^n
    \quad\iff\quad
    SDUy\in\Delta\mathbb{Z}^n
    \quad\iff\quad
    DUy\in\Delta\mathbb{Z}^n\enspace.
\]
Here, the last equivalence exploits unimodularity of $S$.
Letting $\gamma_i^\top$ denote the $i$\textsuperscript{th} row of $U$, we can further rewrite
\[
    DUy\in\Delta\mathbb{Z}^n
    \quad\iff\quad
    \forall i\in[n]\colon\ \widetilde{m}_i \gamma_i^\top y \in \Delta\mathbb{Z}^n
    \quad\iff\quad
    \forall i\in[n]\colon\ \gamma_i^\top y \equiv 0\pmod*{m_i}\enspace,
\]
where we use $m_i\coloneqq \sfrac{\Delta}{\gcd(\Delta, \widetilde{m}_i)}$.%
\footnote{Note we might have $\widetilde{m}_i=0$ for some $i\in[n]$. In this case, $\gcd(\Delta, \widetilde{m}_i)=\Delta$ and hence $m_i=1$, so the corresponding congruency constraint is always satisfied.}
It is thus left to show $\Delta=\prod_{i=1}^n m_i$.
To this end, consider the composed map
\[
	\Phi \colon \Z^n \overset{U}{\longrightarrow} \Z^n \overset{\pi}{\longrightarrow} \prod_{i=1}^n \quot{\Z}{m_i\Z}\enspace,
\]
where, for $z\in\Z^n$, the first map is defined by $z\mapsto Uz$, and the second is the component-wise projection given by $\pi(z)\coloneqq(\pi_1(z_1),\ldots,\pi_n(z_n))$, where $\pi_i$ denotes the natural projection from $\Z$ to $\inlinequot{\Z}{m_i\Z}$.
As $U$ is unimodular, the first map is an isomorphism of groups.
Furthermore, $\pi$ is a surjective group homomorphism.
By the isomorphism theorem, we thus get an isomorphism
\[
	\prod_{i=1}^n \quot{\Z}{m_i\Z} = \operatorname{im}\Phi \cong \quot{\Z^n}{\ker \Phi}\enspace.
\]
Note that the cardinality of the left-hand side group is $\prod_{i=1}^n m_i$.
Therefore, we may finish the proof by showing that $\lvert\inlinequot{\Z^n}{\ker\Phi}\rvert = \Delta$.
To this end, observe that $\ker \Phi$ is the set of $y\in \Z^n$ fulfilling the congruency constraints, i.e., $\gamma_i^\top y\equiv 0 \pmod{m_i}$ for all $i\in[n]$.
By the above discussion, this is precisely the set $\{y \in \Z^n\colon H^{-1}y\in \Z^n\} = H\Z^n$.
Consequently, $\lvert\inlinequot{\Z^n}{\ker\Phi}\rvert = \lvert\inlinequot{\Z^n}{H\Z^n}\rvert=\lvert\det(H)\rvert=\Delta$, as desired.
\end{proof}

\section{GCTUF with transposed network constraint matrices}\label{sec:transposedNetworkBaseBlock}

In the setting with a congruency constraint instead of a group constraint, \cite{nagele_2022_congruency} shows that every base block problem with a constraint matrix that is a transposed network matrix can be reduced to a node-weighted minimization problem over a lattice with a congruency constraint,%
\footnote{%
In fact, the proof in~\cite{nagele_2022_congruency} claims a reduction to a submodular minimization problem, but shows the stronger one presented here.%
}
i.e., a problem of the form
\begin{equation}\label{eq:CCSM}
\min \{w(S)\colon S\in\mathcal{L},\,\gamma(S)\equiv r\pmod*{m} \}\enspace,
\end{equation}
where $\mathcal{L}\subseteq 2^N$ is a lattice on some finite ground set $N$, $\gamma\colon N\to\mathbb{Z}$, $r\in\mathbb{Z}$, $m\in\mathbb{Z}_{>0}$, $w\colon N\to\mathbb{R}$, and we use $\gamma(S)\coloneqq \sum_{v\in S}\gamma(v)$ as well as $w(S)\coloneqq \sum_{v\in S}w(v)$.%
\footnote{%
\label{ftn:lattice}%
We recall that a \emph{lattice} $\mathcal{L}\subseteq 2^N$ is a set family such that for any $A,B\in\mathcal{L}$, we have $A\cap B, A\cup B\in\mathcal{L}$.
We assume such a lattice to be given by a compact encoding in a directed acyclic graph $H$ on the vertex set $N$ such that $X\subseteq N$ is an element of the lattice if and only if $\delta_H^-(X)=\emptyset$ (cf.~\cite[Section 10.3]{groetschel_1993_geometric}). Here, as usual, in a digraph $G=(V,A)$ and for $X\subseteq V$, we denote by $\delta^+(X)$ and $\delta^-(X)$ the arcs in $A$ leaving and entering $X$, respectively. Moreover, we write $\delta^\pm(v)\coloneqq\delta^\pm(\{v\})$ for $v\in V$.%
}
Being a special case of congruency-constrained submodular minimization, it is known that such problems, and thus the corresponding congruency-constrained TU problems with a transposed network constraint matrix, can be solved in strongly polynomial time for constant prime power moduli $m$, while the case of general constant composite moduli remains open~\cite{nagele_2018_submodular}.
The progress on \GCTUF{}, particularly the reduction to base block feasibility problems through \cref{thm:reductionToBaseBlocksG4} and its generalization (\cref{thm:reductionToBaseBlocksRGCTUF} in \cref{sec:technicalOverview}), motivates studying these reductions and results in the feasibility setting and with a group constraint instead of a congruency constraint, giving rise to the following problem.

\begin{mdframed}[innerleftmargin=0.5em, innertopmargin=0.5em, innerrightmargin=0.5em, innerbottommargin=0.5em, userdefinedwidth=0.95\linewidth, align=center]
	{\textbf{Group-Constrained Lattice Feasibility (\GCLF{}\linkdest{prb:GCLF}):}}
    Let $N$ be a finite set, $\mathcal{L}\subseteq 2^N$ a lattice, $(G,+)$ a finite abelian group, $\gamma \colon N \rightarrow G$, $r\in G$.
    The task is to find $X\in\mathcal{L}$ with $\gamma(X) = r$, or decide infeasibility.
\end{mdframed}
We observe that the reduction in~\cite{nagele_2022_congruency} from congruency-constrained TU problems with transposed network constraint matrices to problems of the form given in~\eqref{eq:CCSM} extends to the group-con\-strained case.
In particular, we obtain the following result in the feasibility setting.
(See \cref{sec:transposedNetworkBB} for some details.)%

\begin{proposition}\label{prop:GCTUFtoGCLF}
Let $G$ be a finite abelian group. Any \GCTUF{} problem with group $G$ and a constraint matrix that is a transposed network matrix can in strongly polynomial time be reduced to a \GCLF{} problem with group $G$.
\end{proposition}

Thus, it remains to study \GCLF{} problems.
Interestingly, for the pure feasibility question, we can circumvent the barriers present in the optimization setting, and obtain the following result through a concise argument.

\begin{theorem}\label{thm:GCLF}
Let $G$ be a finite abelian group.
\GCLF{} problems with group $G$ can be solved in strongly polynomial time.
\end{theorem}

Clearly, \cref{prop:GCTUFtoGCLF} and \cref{thm:GCLF} together imply \cref{thm:transpose-bb}.
The main observation towards a proof of \cref{thm:GCLF} is the following elementary lemma.

\begin{lemma}\label{lem:elementary}
Let $G$ be a finite abelian group, and let $\gamma_1,\ldots,\gamma_\ell\in G$.
If $\ell\geq|G|$, then there is a non-empty subset $I\subseteq [\ell]$ such that $\sum_{i\in I} \gamma_i=0$.
\end{lemma}

\begin{proof}
Either $s_i\coloneqq \sum_{j\leq i}\gamma_j=0$ for some $i\in [\ell]$, or there exist $i<j$ with $s_i=s_j$; hence $I=[i]$ or $I=\{i+1,\ldots,j\}$, respectively, has the desired properties.
\end{proof}

To prove \cref{thm:GCLF}, we work with a representation of the lattice $\mathcal{L}$ through an acyclic digraph $H$ (see \cref{ftn:lattice}).
We exploit that every $X\in\mathcal{L}$ is uniquely defined by the subset $C_X\coloneqq\{x\in X\colon \delta^+(x)\subseteq\delta^+(X)\}$.

\begin{proof}[Proof of \cref{thm:GCLF}]
We claim that if the given \GCLF problem is feasible, there is a feasible $X$ with $|C_X| < |G|$.
If so, we obtain an efficient procedure for \GCLF with group $G$ through enumerating all such $C_X$ and checking if $\gamma(X)=r$.
To prove the claim, assume for contradiction that it is wrong, and let $X\in\mathcal{L}$ be minimal with $\gamma(X) = r$.
Then $|C_X| \geq |G|$, and applying \cref{lem:elementary} to $C_X$ gives a non-empty subset $Y\subseteq C_X$ with $\gamma(Y) = 0$.
Thus, $X\setminus Y$ is a strictly smaller lattice element with $\gamma(X\setminus Y) = \gamma(X) - \gamma(Y) = \gamma(X) = r$, a contradiction.
\end{proof}

\section{Approaching GCTUF problems and a proof of \cref{thm:reductionToBaseBlocksG4}}%
\label{sec:technicalOverview}

In order to tackle \GCTUF{} problems, following ideas from~\cite{nagele_2022_congruency}, we introduce a hierarchy of slightly relaxed \GCTUF{} problems by weakening the group constraint.

\begin{mdframed}[innerleftmargin=0.5em, innertopmargin=0.5em, innerrightmargin=0.5em, innerbottommargin=0.5em, userdefinedwidth=0.95\linewidth, align=center]
{\textbf{\boldmath$R$-Group-Constrained TU Feasibility (\RGCTUF{}\linkdest{prb:R-GCTUF}):}}
\sloppy Let $T\in\{-1, 0, 1\}^{k\times n}$ be TU, $b\in\mathbb{Z}^k$, let $(G,+)$ be a finite abelian group, $\gamma\in G^{n}$ and $R\subseteq G$. The task is to show infeasibility or find a solution of
$$
Tx \leq b,\ \gamma^\top x \in R,\ x\in\mathbb{Z}^n \enspace.
$$
\end{mdframed}
Here, we typically call $R$ the set of target elements. The above setup allows us to measure progress between \GCTUF{} (the case of $|R|=1$) and an unconstrained IP with TU constraint matrix (captured by setting $R = G$).
In particular, the difficulty of an \RGCTUF{} problem increases as the size of $R$, i.e., the number of target elements, decreases.
The main parameter capturing this hardness is the \emph{depth} $d\coloneqq|G|-|R|$ of the problem.
We show the following generalization of \cref{thm:GCTUFmod4}.

\begin{theorem}\label{thm:RGCTUF}
Let $G$ be a finite abelian group. There is a strongly polynomial ran\-do\-mized algorithm solving \RGCTUF{} problems with group $G$ and $|G|- |R|\leq 3$.
\end{theorem}

Our approach exploits uses Seymour's decomposition theorem for TU matrices.
To state this result, we first introduce the additional notions of a $3$-sum of matrices, and pivoting operations.

\begin{definition}[$3$-sum]
Let $A\in\mathbb{Z}^{k_A\times n_A}$, $B\in\mathbb{Z}^{k_B\times n_B}$, $e\in\mathbb{Z}^{k_A}$, $f\in\mathbb{Z}^{n_B}$, $g\in\mathbb{Z}^{k_B}$,  $h\in\mathbb{Z}^{n_A}$.
The \emph{$3$-sum} of $\begin{psmallmatrix}
A & e & e \\ h^\top & 0 & 1
\end{psmallmatrix}$ and $\begin{psmallmatrix}
0 & 1 & f^\top \\ g & g & B
\end{psmallmatrix}$ is
$
\begin{psmallmatrix}
A & e & e \\ h^\top & 0 & 1
\end{psmallmatrix} \ksum[3] \begin{psmallmatrix}
0 & 1 & f^\top \\ g & g & B
\end{psmallmatrix} \coloneqq \begin{psmallmatrix}
A & ef^\top \\ gh^\top & B
\end{psmallmatrix}
$.
\end{definition}

\begin{definition}[Pivoting]
	\label{def:pivoting}
	Let $C\in\mathbb{Z}^{k\times n}$, $p\in\mathbb{Z}^n$, $q\in\mathbb{Z}^k$, and $\varepsilon\in\{-1,1\}$. The matrix obtained from pivoting on $\varepsilon$ in $T \coloneqq \begin{psmallmatrix}
		\varepsilon & p^\top \\ q & C
	\end{psmallmatrix}$, i.e., pivoting on the element $T_{11}$ of $T$, is
	$
	\pivot[11](T) \coloneqq \begin{psmallmatrix}
		-\varepsilon & \varepsilon p^\top \\ \varepsilon q & C - \varepsilon q p^\top
	\end{psmallmatrix}%
	$.
	More generally, $\pivot[ij](T)$ for indices $i$ and $j$ such that $T_{ij}\in\{-1,1\}$ is obtained from $T$ by first permuting rows and columns such that the element $T_{ij}$ is permuted to the first row and first column, then performing the above pivoting operation on the permuted matrix, and finally reversing the row and column permutations.
\end{definition}

With this notation at hand, we can state Seymour's TU decomposition theorem as follows.

\begin{theorem}[Seymour's TU decomposition]\label{thm:TUdecomp}
Let $T\in\mathbb{Z}^{k\times n}$ be a totally unimodular matrix. Then, one of the following cases holds.
\begin{enumerate}
\item\label{thmitem:TUdecomp_netw} $T$ or $T^\top$ is a network matrix.

\item\label{thmitem:TUdecomp_const} $T$ is, possibly after iteratively applying the operations of
\begin{itemize}
\item deleting a row or column with at most one non-zero entry,
\item deleting a row or column that appears twice or whose negation also appears in the matrix, and
\item changing the sign of a row or column,
\end{itemize}
equal to one of
\begin{equation*}
\begin{psmallmatrix*}[r]
 1 & -1 &  0 &  0 & -1 \\
-1 &  1 & -1 &  0 &  0 \\
 0 & -1 &  1 & -1 &  0 \\
 0 &  0 & -1 &  1 & -1 \\
-1 &  0 &  0 & -1 &  1
\end{psmallmatrix*}
\quad\text{and}\quad
\begin{psmallmatrix}
 1 &  1 &  1 &  1 &  1 \\
 1 &  1 &  1 &  0 &  0 \\
 1 &  0 &  1 &  1 &  0 \\
 1 &  0 &  0 &  1 &  1 \\
 1 &  1 &  0 &  0 &  1
\end{psmallmatrix}\enspace.
\end{equation*}

\item\label{thmitem:TUdecomp_pivotsum} $T$ can, possibly after row and column permutations and pivoting once, be decomposed into a $3$-sum of totally unimodular matrices with $n_A, n_B\geq 2$.
\end{enumerate}
Additionally, we can in time $\mathrm{poly}(n)$ decide which of the cases holds and determine the involved matrices.
\end{theorem}

We remark that typically, a $3$-sum of the form
$\begin{psmallmatrix}
A & ef^\top \\ gh^\top & B
\end{psmallmatrix}$
would be called a \emph{$1$-} or \emph{$2$-sum} if both or one of the off-diagonal blocks $ef^\top$ and $gh^\top$ were zero, respectively.
However, as we treat those special cases in the same way as $3$-sums, there is no need for us to further distinguish between them.
Generally, we refer to matrices covered by \cref{thmitem:TUdecomp_netw,thmitem:TUdecomp_const} of \cref{thm:TUdecomp} as \emph{base block} matrices.
By porting results on congruency-constrained base block problems of~\cite{nagele_2022_congruency} to the group-constrained setting and combining them with our new \cref{thm:transpose-bb}, it follows that \GCTUF{} problems can be solved in strongly polynomial time if the constraint matrix is a base block matrix.
The potential pivoting step in \cref{thmitem:TUdecomp_pivotsum} of \cref{thm:TUdecomp} can also be handled by extending a result from~\cite{nagele_2022_congruency} to the group setting.
For the sake of completeness, we comment on how to extend the arguments from~\cite{nagele_2022_congruency} for base block constraint matrices or pivot steps in \cref{sec:BBreduction,sec:pivot}, respectively.
Showing how to deal with \RGCTUF problems with constraint matrices that are $3$-sums will lead to the following generalization of \cref{thm:reductionToBaseBlocksG4} which, in combination with the aforementioned results on base block problems, immediately implies \cref{thm:RGCTUF}.
We devote the rest of this section to a discussion of its proof.

\begin{theorem}\label{thm:reductionToBaseBlocksRGCTUF}
Let $G$ be a finite abelian group and $\ell\in\mathbb{Z}_{\geq 1}$ with $\ell\geq
|G|-3$.
Given an oracle for solving
{%
base block \RGCTUF problems with group $G$ and any $R\subseteq G$ with $|R|\geq \ell$, %
}%
we can solve
{%
\RGCTUF problems with group $G$ and $R\subseteq G$ with $|R| \geq\ell$ %
}%
in strongly polynomial time with strongly polynomially many calls to the oracle.
\end{theorem}

\subsection{Reducing to a simpler problem when the target elements form a union of cosets}\label{sec:cosetReduction}

If $R$, the set of target elements, is a union of cosets of the same non-trivial proper subgroup $H$ of $G$ (i.e., it is of the form $R=\bigcup_{i=1}^k (g_i+H)$ for some $g_1,\ldots,g_k\in G$, or equivalently, $R=R+H$), we can directly reduce to a simpler problem.
We formalize this in the following lemma.
\begin{lemma}\label{lem:cosetReduction}
Assume we are given an \RGCTUF{} problem
$$
Tx \leq b,\ \gamma^\top x \in R,\ x\in\mathbb{Z}^n
$$
such that $R=R+H$ for a non-trivial proper subgroup $H$ of $G$.
Then, the set of feasible solutions of the given \RGCTUF{} problem is invariant under replacing $G$ by the quotient group $\widehat G = \inlinequot{G}{H}$, $R$ by $\widehat R = \inlinequot{R}{H}$, and $\gamma$ by its image $\widehat \gamma\in \widehat{G}^n$ under the quotient map.
\end{lemma}

\begin{proof}
Let $P$ denote the original \RGCTUF{} problem, and let $\widehat P$ denote the modified one.
The inequality system $Tx\leq b$ is the same in $P$ and $\widehat P$, hence it is enough to show that $\gamma^\top x \in R$ if and only if $\widehat \gamma^\top x \in \widehat R$.

To this end, first note that for $x\in\mathbb{Z}^n$, $\gamma^\top x \in R$ immediately implies $\widehat{\gamma}^\top x \in \widehat R$ by definition.
For the other direction, assume $x\in \mathbb{Z}^n$ satisfies $\widehat \gamma^\top x \in R$.
Then, by definition of $\widehat \gamma$ and $\widehat R$, we know that there is an element $h\in H$ such that $\gamma^\top x + h \in R$.
Then $\gamma^\top x \in R - h = R$, as desired.
\end{proof}

The depth of the new problem given by \cref{lem:cosetReduction} in the corresponding hierarchy is $\widehat d=|\inlinequot{G}{H}| - |\inlinequot{R}{H}| = \frac{|G| - |R|}{|H|} < |G| - |R|$, so we indeed end up with a simpler problem in that respect.
Since the existence of such a subgroup $H$ can be checked efficiently (given that $G$ has constant size), we can always and in constant time determine upfront whether the \RGCTUF{} problem at hand is reducible using \cref{lem:cosetReduction}, and if so, reduce it to a simpler \RGCTUF{} problem.
Thus, for the rest of this section, we assume $R$ is not a union of cosets.
This assumption allows us to apply a special case of the Cauchy-Davenport theorem that holds despite the fact that the group order may not be prime. We refer to \cref{lemma:CD-replacement} for details.

\subsection{Decomposing the problem}\label{sec:highlevel}
We now focus on an \RGCTUF problem with a constraint matrix $T$ that can be decomposed into a $3$-sum of the form $T=\begin{psmallmatrix}A & ef^\top \\ gh^\top & B\end{psmallmatrix}$.
The decomposition allows for splitting $x$, $b$, and $\gamma$ into two parts accordingly, giving the equivalent formulation
\begin{equation}\label{eq:structured-problem}
\begin{pmatrix}
A & ef^\top \\ gh^\top & B
\end{pmatrix} \cdot \begin{pmatrix}
x_A \\ x_B
\end{pmatrix} \leq \begin{pmatrix}
b_A \\ b_B
\end{pmatrix}\enspace,\quad
\gamma_A^\top x_A + \gamma_B^\top x_B \in R\enspace,\quad
\begin{aligned}
x_A &\in \mathbb{Z}^{n_A} \\
x_B &\in \mathbb{Z}^{n_B}
\end{aligned}\enspace.
\end{equation}
In the inequality system, the variables $x_A$ and $x_B$ interact only through the rank-one blocks $ef^\top$ and $gh^\top$.
Fixing values of $\alpha\coloneqq f^\top x_B$ and $\beta\coloneqq h^\top x_A$ allows for rephrasing~\eqref{eq:structured-problem} through the following two almost independent problems
\begin{equation}%
\label{eq:A-B-problem}%
\begin{minipage}{.25\linewidth}
\setlength{\abovedisplayskip}{0pt}%
$$\begin{aligned}
A x_A &\leq b_A - \alpha e\\
h^\top x_A  &= \beta \\
x_A  &\in \mathbb{Z}^{n_A}
\end{aligned}$$
\end{minipage}
\quad\text{and}\quad
\begin{minipage}{.25\linewidth}
\setlength{\abovedisplayskip}{0pt}%
$$\begin{aligned}
B x_B &\leq b_B - \beta g\\
f^\top x_B  &= \alpha \\
x_B  &\in \mathbb{Z}^{n_B}
\end{aligned}$$
\end{minipage}~,
\enspace
\end{equation}%
where we seek to find solutions $x_A$ and $x_B$ such that their corresponding group elements $r_A\coloneqq \gamma_A^{\top} x_A$ and $r_B\coloneqq \gamma_B^{\top} x_B$, respectively, satisfy $r_A + r_B \in R$.
Hence, this desired relation between the target elements $r_A$ and $r_B$ is the only dependence between the two problems once $\alpha$ and $\beta$ are fixed.
We assume without loss of generality that $A$ has no fewer columns than $B$, and refer to the problem on the left as the \emph{A-problem}, and the problem on the right as the \emph{B-problem}.
We denote by $\Pi$ the set of all $(\alpha, \beta)\in \mathbb{Z}^{2}$ such that both the $A$- and $B$-problem are feasible. (Note that both problems are described through a TU constraint matrix; hence, feasibility can be checked efficiently.)
Also, for $(\alpha,\beta)\in \Pi$, let $\pi_A(\alpha,\beta) \subseteq G$ be all group elements $r_A\in G$ for which there is a solution $x_A$ to the $A$-problem with $\gamma^\top x_A = r_A$, and define $\pi_B$ analogously. We refer to $\pi_A$ and $\pi_B$ as \emph{patterns}.
Hence, \eqref{eq:structured-problem} is feasible if and only if there is a pair $(\alpha,\beta)\in \Pi$ such that, for some $r_A\in \pi_A(\alpha,\beta)$ and $r_B \in \pi_B(\alpha,\beta)$, we have $r_A+r_B\in R$.
Thus, patterns contain all information needed to decide feasibility.

Using techniques from~\cite{nagele_2022_congruency}, we can restrict our search for feasible solutions to a constant-size subset $\widehat{\Pi} \subseteq \Pi$.
More precisely, one can show the following (we give more details in \cref{sec:pf_boundedPattern}).
\begin{lemma}\label{lemma:boundedPatternShape}
	One can in strongly polynomial time find $\ell_i, u_i\in\mathbb{Z}$ for $i\in\{0,1,2\}$, with $u_i-\ell_i\leq d$  such that
	\begin{equation}\label{eq:patternShape}
		\widehat{\Pi} \coloneqq \left\{(\alpha, \beta)\in\mathbb{Z}^2 \colon \ell_0 \leq \alpha+\beta \leq u_0, \ell_1 \leq \alpha \leq u_1, \ell_2\leq \beta \leq u_2\right\}
	\end{equation}
    satisfies $\Pinar\subseteq\Pi$, and if~\eqref{eq:structured-problem} is feasible, then there is a pair $(\alpha,\beta)$ in $\Pinar$
	for which there is a solution $x_A$ to the $A$-problem and a solution $x_B$ to the $B$-problem with $\gamma^\top x_A + \gamma^\top x_B \in R$.
\end{lemma}

Therefore, the challenges lie less in the size of $\Pi$, but rather in how to obtain information on the sets $\pi_A(\alpha,\beta)$ and $\pi_B(\alpha,\beta)$ for pairs $(\alpha,\beta)\in \Pi$.
Opposed to previous techniques, which almost solely focused on $\pi_B$, we investigate both $\pi_A$ and $\pi_B$ and their interplay---see \cref{sec:patterns}.

As $B$ has at most half the columns of the constraint matrix $T$ of the original \RGCTUF problem~\eqref{eq:structured-problem}, we can afford (runtime-wise) to recursively call our algorithm multiple times on the $B$-problem for different targets $R_B$ of the same depth $d=|G|-|R|$ as the original problem, i.e., with $|R_B|=|R|$.
(We refrain from using larger depths, as \GCTUF become harder with increasing depth.)
This allows us to compute a set $\bar{\pi}_B(\alpha,\beta)\subseteq \pi_B(\alpha,\beta)$ of size $|\bar{\pi}_B(\alpha,\beta))|=\min\{d+1,\pi_B(\alpha,\beta)\}$.
Indeed, we can start with $\bar{\pi}_B(\alpha,\beta)=\emptyset$ and, as long as $|\bar{\pi}_B(\alpha,\beta)| < \min \{d+1,\pi_B(\alpha,\beta)\}$, we solve an \RGCTUF[R_B] $B$-problem (i.e., we look for a $B$-problem solution $x_B$ with $\gamma^\top x_B \in R_B$) with $R_B = G\setminus \bar{\pi}_B(\alpha,\beta)$ being a set of size at least $|G|-d$.
If $R_B\cap \pi_B(\alpha,\beta)\neq \emptyset$, then we find an element in $R_B\cap \pi_B(\alpha,\beta)$ that can be added to $\bar{\pi}_B(\alpha,\beta)$ and we repeat; otherwise, $R_B\cap \pi_B(\alpha,\beta)=\emptyset$ and we know that we computed $\bar\pi_B\ab=\pi_B\ab$.

To the contrary, note that the $A$-problem may be almost as big as the original \GCTUF problem (possibly with just two fewer columns).
Hence, here we cannot afford (runtime-wise) a similar computation as for the $B$-problem.
However, we can afford to solve multiple \RGCTUF[R_A] $A$-problems of smaller depth, i.e., $|R_A|>|R|$, because the runtime decreases significantly with decreasing depth.
By using the same approach as in the $B$-problem, but with sets $R_A$ of size $|R_A|\geq|R|+1$, we obtain a set $\bar{\pi}_A(\alpha,\beta)\subseteq \pi_A(\alpha,\beta)$ of size $|\bar{\pi}_A(\alpha,\beta)|=\min\{d,\pi_A(\alpha,\beta)\}$.

Let us next take a closer look at patterns.
Fix some $\ab \in \Pi$ and let $\pi_A \ab = \{r_A^1,\ldots,r_A^{\ell_A}\}$ for some $\ell_A \geq 1$ and pairwise different $r_A^i\in G$, and let $x_A^{1},\ldots, x_A^{\ell_A}$ be corresponding solutions of the $A$-problem with $\gamma_A^\top x_A^i=r_A^i$.
Define $\ell_B$, $r_B^i$, and $x_B^i$ analogously.
Observe that if $\ell_A\leq d$ and $\ell_B\leq d+1$, we have $\bar\pi_X\ab = \pi_X\ab$ for both $X\in\{A,B\}$.
Hence, we can compute all feasible group elements and check explicitly whether $r_A^i+r_B^j\in R$ for some $i\in[\ell_A]$ and $j\in[\ell_B]$, i.e., whether a solution exists.
If $\ell_B\geq d+1$, we can (independently of $\ell_A$) even show that there always exists a feasible solution, and we can also find one:
Indeed, we can compute $d+1$ solutions $x^i\coloneqq (x^1_A, x^i_B)$ with pairwise different sums $r^1_A + r^i_B\in G$, at least one of which must satisfy $r^1_A + r^i_B \in R$.
If $\ell_A\geq d$ and $\ell_B\geq 2$, we can argue similarly:
We show that among any $d$ elements of $\bar{\pi}_A \ab$, and any two elements of $\bar{\pi}_B \ab$ (which we can compute), there is a pair $r^i_A, r^j_B$ with $r^i_A + r^j_B \in R$.
Note that while for groups of prime order this can be shown via the Cauchy-Davenport theorem, the above result does not hold in general.
We show, however, that as long as $R$ is not a union of cosets in $G$, we can recover the implication (cf. \cref{sec:cosetReduction} for why this assumption is legit).

\begin{lemma}\label{lemma:CD-replacement}
Let $G$ be a finite abelian group, and let $R \subseteq G$ be such that $R \neq R+H$ for any non-trivial subgroup $H$ of $G$.
Then, for any subsets $X,Y\subseteq G$ with $|X| = |G|-|R|$ and $|Y|\geq 2$, we have $(X+Y)\cap R \neq \emptyset$.
\end{lemma}
\begin{proof}
Let $b_1, b_2 \in Y$ with $b_1\neq b_2$, and set $h= b_1-b_2$. Assume $\left(X + Y\right)\cap R = \emptyset$. Then $|X| = |G| - |R|$ implies $|X + Y| =|X|$.
Thus, $X + b_1 = X + b_2$ and hence  $X = X + h$.
Iterating gives $X = X + \langle h \rangle$, where $\langle h \rangle$ denotes the subgroup generated by $h$. As $R = G \setminus (X + b_1)$, we get $R = R + \langle h \rangle$, a contradiction.
\end{proof}

The following observation summarizes the above discussion.

\begin{observation}\label{obs:immediateConclusion}
Let $\ab \in \widehat{\Pi}$.
If $|\bar{\pi}_A \ab|\leq d-1$ or $|\bar{\pi}_B \ab|\geq2$, we can immediately determine whether a feasible solution to the original \RGCTUF problem exists for such $\ab$, and if so, obtain one by combining solutions computed for the $A$- and $B$-subproblem when determining $\bar\pi_A$ and $\bar\pi_B$.
\end{observation}
Thus, the only case in which we cannot immediately check whether a feasible solution exists for some $\ab$, is when $\ell_B=1$ and $\ell_A \geq d+1$ (which imply $|\bar{\pi}_A \ab| = d$ and $|\bar{\pi}_B \ab|= 1$).
This is the only case where we may have $\left(\pi_A \ab + \pi_B \ab \right)\cap R \neq \emptyset$ but $\left(\bar{\pi}_A \ab + \bar{\pi}_B \ab \right)\cap R = \emptyset$, in which case we say that $\ab$ contains a \emph{hidden solution}.

\subsection{\boldmath New insights towards overcoming previous barriers for $d=3$}%
\label{sec:patterns}

We now describe how our new techniques allow for overcoming barriers restricting previous approaches to depth $d=2$.
Recall that we focus on a constant size subset $\Pinar$ as defined in \eqref{eq:patternShape}.
We call sets of this form, for any choice of $\ell_i$ and $u_i$, \emph{pattern shapes}, and denote by
\begin{equation}\label{eq:directions}
    \mathcal{D}\coloneqq \left\{
\pm\begin{psmallmatrix} 1 \\ 0 \end{psmallmatrix},
\pm\begin{psmallmatrix} 0 \\ 1 \end{psmallmatrix},
\pm\begin{psmallmatrix} 1 \\ -1 \end{psmallmatrix}\right\}
\end{equation}
the possible edge directions of $\operatorname{conv}(\Pinar)$.
Focusing on $\Pinar$ allows for efficiently computing $\bar\pi_X\ab$ for $X\in\{A,B\}$ and all $\ab\in\Pinar$ to the extent discussed earlier.
In order to proceed, we use a structural result from~\cite{nagele_2022_congruency}, called averaging, that allows us to relate solutions---and thus elements of $\pi_X$---across different $\ab$.
Despite being true in more generality, the exposition here requires the following special case only.

\begin{proposition}[{\cite[special case of Lemma~5.3]{nagele_2022_congruency}}]\label{prop:averaging}
Consider an \RGCTUF{} problem as described in~\eqref{eq:structured-problem}.
 Let $X \in \{A,B\}$, $v\in \mathcal{D}$, and $\ab\in \Pinar$ with $\ab + 2v \in \Pinar$. Given a solution $x_1$ of the $X$-problem for $\ab$ and, similarly, $x_2$ for $\ab+2v$, there are solutions $x_3,x_4$ for the $X$-problem for $\ab + v$ such that $x_1+x_2 = x_3+x_4$.
\end{proposition}

We remark that the proof of the above result for congruency-constrained problems given in~\cite{nagele_2022_congruency} only exploits that congruency-constraints are linear constraints; therefore, the result carries over to group-constraints seamlessly.

In previous approaches for depth $d=2$, it was enough to only compute a single element from $\pi_A$ (e.g., by solving the $A$-problem after dropping the group constraint).
Concretely, consider patterns of the shape as given in \cref{fig:barrier}. For $d=2$, \cref{prop:averaging} can be used to show that, if there is a hidden feasible solution for $\ab=(0,0)$ or $\ab=(2,0)$, then there must also be a feasible solution for $\ab=(1,0)$.
The example in \cref{fig:barrier} shows that this is no longer true if the depth $d$ exceeds~$2$, as only $\ab=(0,0)$ admits a feasible solution.

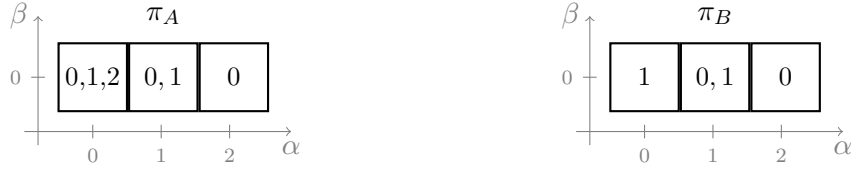
\begin{figure}[!ht]
\centering
\begin{subfigure}{0.45\textwidth}
\centering
\begin{tikzpicture}[scale=0.9]

\begin{scope}[every node/.style={patternSquare}]
\node (l) at (0.5,0.5) {$0,\! 1,\! 2$};
\node[right=0cm of l] (m) {$0,1$};
\node[right=0cm of m] (r) {$0$};
\end{scope}

\node[above=0.1cm of m, font=\normalsize] {$\pi_A$};

\begin{scope}[every path/.append style={gray}]
\draw[->] (-0.5, -0.3) -- node[pos=1, below] {$\alpha$} (3.4,-0.3);
\draw[->] (-0.3, -0.5) -- node[pos=1, left] {$\beta$} (-0.3,1.4);
\foreach \tick/\num in {0.5/0,1.5/1,2.5/2} {
\draw (\tick, -0.2) -- (\tick, -0.4) node[below, font=\scriptsize]{$\num$};
}
\draw (-0.2, 0.5) -- (-0.4, 0.5) node[left, font=\scriptsize]{$0$};
\end{scope}
\end{tikzpicture} %
\end{subfigure}
\begin{subfigure}{0.45\textwidth}
\centering
\begin{tikzpicture}[scale=0.9]

\begin{scope}[every node/.style={patternSquare}]
\node (l) at (0.5, 0.5) {$1$};
\node[right=0cm of l] (m) {$0, 1$};
\node[right=0cm of m] (r) {$0$};
\end{scope}

\node[above=0.1cm of m, font=\normalsize] {$\pi_B$};

\begin{scope}[every path/.append style={gray}]
\draw[->] (-0.5, -0.3) -- node[pos=1, below] {$\alpha$} (3.4,-0.3);
\draw[->] (-0.3, -0.5) -- node[pos=1, left] {$\beta$} (-0.3,1.4);
\foreach \tick/\num in {0.5/0,1.5/1,2.5/2} {
\draw (\tick, -0.2) -- (\tick, -0.4) node[below, font=\scriptsize]{$\num$};
}
\draw (-0.2, 0.5) -- (-0.4, 0.5) node[left, font=\scriptsize]{$0$};
\end{scope}
\end{tikzpicture} %
\end{subfigure}
\caption{Possible patterns $\pi_A$ and $\pi_B$ for a problem with group $G=\sfrac{\mathbb{Z}}{4\mathbb{Z}}$. Every square corresponds to a pair $\ab\in\Pinar$, and the numbers in the box indicate elements of $\pi_A\ab$ and $\pi_B\ab$, respectively. For $R=\{3\}$, there is a feasible solution with $\ab=(0,0)$, but this cannot be detected without studying $\pi_A$.
}\label{fig:barrier}
\end{figure}

This problem can be circumvented by analyzing the $A$-pattern $\bar{\pi}_A$.
As argued in \cref{sec:highlevel}, if a pair $\ab$ has a hidden solution, then $|\pi_A\ab| \geq d+1$ (and hence $|\bar{\pi}_A\ab| = d$), hence we assume that there exists at least one such pair.
The following result uses averaging (i.e., \cref{prop:averaging}) to show that pairs $\abprime$ adjacent to such a pair $\ab$ containing a hidden solution also have large $\bar{\pi}_A \abprime$.

\begin{lemma}\label{lem:propagating_hidden_solutions}
Let $d\in\{1,2,3\}$, $v \in \mathcal{D}$, and $(\alpha, \beta)\in\Pinar$ such that $|\pi_A\ab|\geq d+1$ and $(\alpha, \beta)+2v\in\Pinar$.
Then $|\bar{\pi}_A(\ab+v)| = d$.
\end{lemma}

\begin{proof}
It is enough to show that $|\pi_A(\ab+v)|\geq d$.
To this end, for each of the at least $d+1$ elements $r\in\pi_A(\alpha,\beta)$, let $x_1^r$ be a corresponding solution of the $A$-problem, and let $x_2$ denote any fixed solution for the $A$-problem on the pair $\ab+2v$.
\cref{prop:averaging} applied to $x_1^r$ and $x_2$ gives solutions $x_3^r$ and $x_4^r$ corresponding to elements $\gamma_A^\top x_3^r,\gamma_A^\top x_4^r\in\pi_A(\ab+v)$ with $\gamma_A^\top x_3^r + \gamma_A^\top x_4^r$ taking at least $d+1$ different values.
Assume for the sake of deriving a contradiction that $|\pi_A(\ab+v)|\leq d-1$. Then, since the number of different sums of pairs of elements in $\pi_A(\ab+v)$ is bounded by $\binom{d-1}{2} +d - 1 = \sfrac{(d-1)d}{2}< d+1 $ for $d\in\{1,2,3\}$,
this contradicts the above construction.
\end{proof}

\begin{remark}
For depth $d=4$, one can find \GCTUF problems with $G=\sfrac{\mathbb{Z}}{5\mathbb{Z}}$ and patterns that fail to satisfy \cref{lem:propagating_hidden_solutions}; we present one such example in \cref{fig:barrier_example}.
Moreover, we remark that \cref{lem:propagating_hidden_solutions} is the only place in our proofs where we use the assumption that $d = |G|-|R| \leq 3$.
\end{remark}
\begin{figure}[!ht]
\centering
\begin{subfigure}{0.45\textwidth}
\centering
\begin{tikzpicture}[scale=0.9]

\begin{scope}[every node/.style={patternSquare}]
\node[text depth=2.3ex] (l) at (0.5,0.5) {$0, 1,$ $2,\! 3,\! 4$};
\node[text depth=2.3ex, right=0cm of l] (m) {$0, 1,$ $2$};
\node[right=0cm of m] (r) {$0$};
\end{scope}

\node[above=0.1cm of m, font=\normalsize] {$\pi_A$};

\begin{scope}[every path/.append style={gray}]
\draw[->] (-0.5, -0.3) -- node[pos=1, below] {$\alpha$} (3.4,-0.3);
\draw[->] (-0.3, -0.5) -- node[pos=1, left] {$\beta$} (-0.3,1.4);
\foreach \tick/\num in {0.5/0,1.5/1,2.5/2} {
\draw (\tick, -0.2) -- (\tick, -0.4) node[below, font=\scriptsize]{$\num$};
}
\draw (-0.2, 0.5) -- (-0.4, 0.5) node[left, font=\scriptsize]{$0$};
\end{scope}
\end{tikzpicture}
\end{subfigure}
\begin{subfigure}{0.45\textwidth}
\centering
\begin{tikzpicture}[scale=0.9]

\begin{scope}[every node/.style={patternSquare}]
\node (l) at (0.5, 0.5) {$1$};
\node[right=0cm of l] (m) {$0, 1$};
\node[right=0cm of m] (r) {$0$};
\end{scope}

\node[above=0.1cm of m, font=\normalsize] {$\pi_B$};

\begin{scope}[every path/.append style={gray}]
\draw[->] (-0.5, -0.3) -- node[pos=1, below] {$\alpha$} (3.4,-0.3);
\draw[->] (-0.3, -0.5) -- node[pos=1, left] {$\beta$} (-0.3,1.4);
\foreach \tick/\num in {0.5/0,1.5/1,2.5/2} {
\draw (\tick, -0.2) -- (\tick, -0.4) node[below, font=\scriptsize]{$\num$};
}
\draw (-0.2, 0.5) -- (-0.4, 0.5) node[left, font=\scriptsize]{$0$};
\end{scope}
\end{tikzpicture}
\end{subfigure}
\caption{Possible patterns $\pi_A$ and $\pi_B$ for a problem with group $G=\sfrac{\mathbb{Z}}{5\mathbb{Z}}$. Every square corresponds to a pair $\ab\in\Pinar$, and the numbers in the box indicate the elements of $\pi_A\ab$ and $\pi_B\ab$, respectively. For $d=4$, \cref{lem:propagating_hidden_solutions} fails to hold for $\ab=(0,0)$ and $v=(1,0)$.}
\label{fig:barrier_example}
\end{figure}
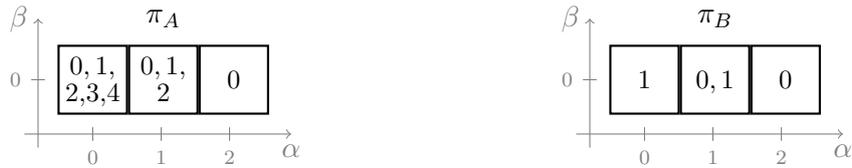

To proceed, we observe that if, on top of the assumption in \cref{lem:propagating_hidden_solutions}, $|\pi_B(\ab+v)| \geq 2$ holds, then \cref{lemma:CD-replacement} guarantees $(\bar\pi_A(\ab+v) + \bar\pi_B(\ab+v)) \cap R \neq \emptyset$, i.e., existence of a feasible solution.
Thus, from now on, we analyze both the $A$- and $B$-patterns in detail, in particular through averaging, to find a pattern constellation as mentioned above, or identify additional properties that allow for direct progress.

\subsection{Analyzing pattern structure}

Before getting to an exhaustive analysis of patterns based on the insights laid out earlier, we introduce notions that will allow us to distinguish patterns from a structural point of view (also see \cref{fig:pattern}).

\begin{definition}\label{def:pattern_types}
Let $\mathcal{D}$ be the possible edge directions of a pattern shape as defined in \eqref{eq:directions}. We call $(\alpha, \beta)\in\Pinar$
an \emph{interior pair} if $(\alpha, \beta) + v \in \Pinar$ for all $v\in \mathcal{D}$,
a \emph{border pair} if $(\alpha, \beta) \pm v \in \Pinar$ for exactly two $v\in \mathcal{D}$, and
a \emph{vertex pair} if it is not an interior or border pair.
\end{definition}

\begin{figure}[ht]%
\centering%
\begin{tikzpicture}[scale=0.6]
\clip (-1.5, -1) rectangle (4.5,4.5);
\begin{scope}[every path/.append style={thick}, every node/.style={patternSquare, minimum width=.6cm, minimum height=.6cm, text width=0.5cm, font={\vphantom{b}}}]
\node at (0.5, 0.5) {$x$};
\node at (1.5, 0.5) {$b$};
\node at (2.5, 0.5) {$x$};
\node at (0.5, 1.5) {$b$};
\node at (1.5, 1.5) {$i$};
\node at (2.5, 1.5) {$x$};
\node at (0.5, 2.5) {$x$};
\node at (1.5, 2.5) {$x$};
\end{scope}

\draw[->] (-0.75, -0.5) -- node[pos=1, above right] {$\alpha$} (3.5,-0.5);
\draw[->] (-0.5, -0.75) -- node[pos=1, above left] {$\beta$} (-0.5,3.5);
\foreach \tick in {0.5,1.5,2.5} {
\draw (\tick, -0.6) -- (\tick, -0.4);
\draw (-0.6, \tick) -- (-0.4, \tick);
}
\end{tikzpicture}
\caption{A pattern shape with interior, border, and vertex pairs (marked $i$, $b$, and $x$, respectively).}
\label{fig:pattern}
\end{figure}
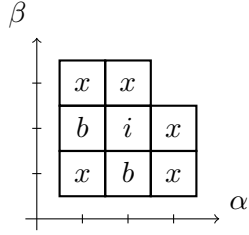

Note that for a border pair $\ab$, due to symmetry, the two directions $v\in \mathcal{D}$ satisfying $\ab\pm v\in\Pinar$ will always be antiparallel, i.e., $v$ and $-v$ for some $v\in\mathcal{D}$.
Next, we summarize results of~\cite{nagele_2022_congruency} that we reuse here.
We remark that these results were proved in a congruency-constrained setting, but translate to problems with group constraints straightforwardly.
For the sake of completeness, we comment on how to adapt the proofs in \cref{sec:pf_blackbox}.

\begin{theorem}[\cite{nagele_2022_congruency}]\label{thm:blackbox}
Consider an \RGCTUF{} problem as described in~\eqref{eq:structured-problem}.
\begin{enumerate}
\item \label{thmitem:next_to_two} If there is some $\ab\in\Pinar$ with $|\bar\pi_B\ab| \geq 2$, then for each $\abprime \in\Pinar$, there exists $v\in\mathcal{D}\cup\{0\}$ such that $|\bar\pi_B(\abprime+v)| \geq 2$.
If in addition, $\Pinar$ contains an interior pair, then for each $\abprime\in\Pinar$, $v$ can be chosen such that we additionally have $\abprime +2v \in \Pinar$.

\item \label{thmitem:linear_pattern} If $|\pi_B\ab|= 1$ for all $\ab\in\Pinar$, or $\Pinar$ only contains vertex pairs and there are no solutions for $\ab\in\Pinar$ with $|\bar\pi_B\ab|\geq 2$, then the problem can be reduced to a single \RGCTUF problem with the same group $G$ and at the same depth $d$, and strictly fewer variables.
\end{enumerate}
\end{theorem}

The two statements in the above theorem serve a complimentary purpose: While \cref{thmitem:linear_pattern} allows for direct progress (by reducing the number of variables), particularly in the case where no $\ab\in\Pinar$ satisfies $|\pi_B\ab|\geq2$, \cref{thmitem:next_to_two} shows that whenever such $\ab$ are present, then they are, in a certain sense, well spread over the pattern.
We will exploit the latter in combination with \cref{lem:propagating_hidden_solutions}.

To formally analyze patterns, based on \cref{obs:immediateConclusion}, we may assume that we face an \RGCTUF problem for which $\Pinar$ contains at least one $\ab$ such that $|\pi_A \ab|\geq d+1$ and $|\pi_B \ab|=1$.
Starting from there, we distinguish four different types of pattern structure as follows:
\begin{enumerate*}[label=(\Roman*), ref={type~\Roman*}]
\item\label{item:type1} $|\pi_B\ab|=1$ for all $\ab\in\Pinar$, or this is not the case and
\item\label{item:type2} $\Pinar$ has an interior pair, or
\item\label{item:type3} $\Pinar$ has no interior but border pairs, or
\item\label{item:type4} $\Pinar$ has only vertex pairs.
\end{enumerate*}
The remainder of this section is devoted to presenting how to achieve progress in each of these four cases.

\subsubsection*{Pattern structure of \ref{item:type1}}
Pattern structure of \ref{item:type1} is covered by \cref{thm:blackbox}~\ref{thmitem:linear_pattern}, which allows to reduce the problem to a new \GCTUF problem with same group $G$ and same depth $d$, and at least one variable less, thus allowing to make progress in that respect.

\subsubsection*{Pattern structure of \ref{item:type2}}
For pattern structure of \ref{item:type2}, we argue that if the \RGCTUF{} problem is feasible, then $\left(\bar{\pi}_A + \bar{\pi}_B\right)\cap R \neq \emptyset$.
More precisely, we show that there must exist $\ab \in \Pinar$ with  $|\bar{\pi}_A \ab| = d$ and $|\bar{\pi}_B \ab| \geq 2$.
This then implies the desired result by \cref{lemma:CD-replacement}.
Concretely, assume that there exist $(\alpha', \beta')\in\Pinar$ containing a hidden solution.
Then, since $\Pinar$ contains an interior pair, and there exists $\ab \in \Pinar$ with $|\pi_B \ab| \geq 2$,
by \cref{thm:blackbox}~\ref{thmitem:next_to_two} there exists
$v\in\mathcal{D}$ such that $(\alpha', \beta') + 2v \in \Pinar$ and $|\pi_B((\alpha', \beta') + v)|\geq 2$.
As \cref{lem:propagating_hidden_solutions} implies that $|\bar\pi_A((\alpha', \beta') + v)| = d$, it follows by \cref{lemma:CD-replacement} that
$\left(\bar{\pi}_A ((\alpha', \beta')+v)  + \bar{\pi}_B ((\alpha', \beta')+v) \right)\cap R \neq \emptyset$; thus we can find a solution at $\abprime+v$.

\subsubsection*{Pattern structure of \ref{item:type3}}
In this case, we show that if $\bar\pi_A + \bar\pi_B$ does not hit the target set $R$, i.e., we fail to find a solution by combining solutions of the $A$- and $B$-problem that we computed recursively, then we can reduce to a smaller pattern shape $\Pi'$, and recurse.
Through such a reduction, we will after constantly many steps reach pattern structures of \ref{item:type1} or \ref{item:type4}, and therefore achieve progress through the techniques presented for the corresponding type.

To start with, we show the following structural auxiliary result.
We remark that this result can be seen as an implication of~\cite[Lemma~5.12]{nagele_2022_congruency}, but we provide a more direct proof here.

\begin{lemma}\label{claim:4vertexSquares}
Assume that $\Pinar$ does not contain interior pairs. Then $\Pinar$ contains at most four vertex pairs.
\end{lemma}

\begin{proof}
Through shifting, we may assume that $\ell_1 = \ell_2 = 0$, and thus also $\ell_3\geq 0$ (if $\ell_3<0$, we may set it to zero without changing $\Pinar$).
Similarly, we may assume $u_3\leq u_1+u_2$.
Note that if $u_1 - \ell_1 \leq 1$ there are at most four vertex pairs:
At most two pairs $\ab$ may satisfy $\alpha = \ell_1$, and at most two further pairs may have $\alpha = u_1$.
Similarly, we are done if $u_2 - \ell_2 \leq 1$, or $u_3 - \ell_3\leq 1$.
Thus, we assume $u_i - \ell_i\geq 2$ for $i\in \{1,2,3\}$.
Consider the pair
\[
\ab = \begin{cases}
(1, \ell_3) & \text{ if } \ell_3 < u_2\\
(\ell_3-u_2+2, u_2-1) & \text{ if } \ell_3 \geq u_2
\end{cases}\enspace.
\]
By definition, $0<\alpha$, $\beta<u_2$, and $\ell_3 < \ell_3+1 = \alpha+\beta < u_3$.
Because $\ab$ can not be an interior pair, we must either have $\alpha \geq u_1$, or $\beta \leq 0$.
Because $u_1\geq 2+\ell_1=2$, we can only have $\alpha \geq u_1$ in the case $\ell_3\geq u_2$, which implies $\ell_3 + 2 \geq u_1+u_2$.
As also, $\ell_3 + 2 \leq u_3 \leq u_1+u_2$, these inequalities must be tight, implying that there are precisely the three vertex pairs $(u_1, u_2), (u_1-2,u_2)$, and $(u_1, u_2-2)$.
Similarly, because $u_2-1\geq 1+\ell_2\geq 1$, we can only have $\beta\leq 0$ in the case $\ell_3<u_2$, which implies $\ell_3=0$.
Consequently, we must also have $u_2=2$; otherwise $(1,1)$ is an interior pair.
This implies that there are precisely the three vertex pairs $(0,0), (2,0), (0,2)$, and completes the proof.
\end{proof}

With the above at hand, we can achieve the desired progress for pattern structure of \ref{item:type3}.

\begin{lemma}\label{lemma:easy_border}
Consider an \RGCTUF{} instance of the form given in \eqref{eq:structured-problem}, and let the corresponding pattern shape $\Pinar$ be of \ref{item:type3}. Then either $\left(\bar{\pi}_A\ab + \bar{\pi}_B\ab\right)\cap R \neq \emptyset$ for some $\ab\in\Pinar$, or we can in strongly polynomial time find a pattern shape $\Pi' \subsetneq \Pinar$ such that the \RGCTUF{} instance is feasible on $\Pinar$ if and only if it is feasible on $\Pi'$.
\end{lemma}

\begin{proof}
First assume there is a border pair $\ab$ with $|\pi_B\ab| \geq 2$.
Consider the unique constraint in the inequality description of $\Pinar$ that is tight at $\ab$, and assume $\abprime\in\Pinar\setminus\{\ab\}$ is another pair satisfying the same constraint with equality, and such that there is a hidden solution at $\abprime$.
Let $v \in \mathcal{D}$ be the direction pointing from $(\alpha',\beta')$ to $(\alpha,\beta)$.
By applying \cref{prop:averaging} repeatedly, we get $|\bar{\pi}_B((\alpha',\beta') + v)| \geq 2$. Moreover, as $(\alpha,\beta)$ is a border pair, we have that $(\alpha,\beta) + v$ lies in $\Pinar$, and hence so does $(\alpha',\beta') + 2v$.
It then follows from \cref{lem:propagating_hidden_solutions} that $|\bar{\pi}_A( (\alpha', \beta') +v)|= d$, and hence $\left(\bar{\pi}_A ((\alpha', \beta')+v) + \bar{\pi}_B ((\alpha', \beta')+v) \right)\cap R \neq \emptyset$ by \cref{lemma:CD-replacement}.
Thus, if we do not find a solution right away by combining elements from $\bar\pi_A\ab$ and $\bar\pi_B\ab$, then there cannot be a hidden solution anywhere on the tight constraint.
Thus, strengthening the tight constraint by one unit leads to the desired pattern shape $\Pi'\subsetneq \Pinar$.

It is left to study the case where no border pair $\ab\in\Pinar$ satisfies $|\pi_B\ab| \geq 2$.
Since $\Pinar$ has no interior pairs, there must exist a vertex pair $(\alpha', \beta') \in \Pinar$ with $|\bar{\pi}_B(\alpha', \beta')|\geq 2$ (otherwise $\Pinar$ would be of type I).
We now distinguish two cases:

\emph{Case 1: There are at least three vertex pairs with $|\bar{\pi}_B \ab| =1$.}
Combining the above discussion with \cref{claim:4vertexSquares}, it follows that there is exactly one $(\alpha', \beta') \in \Pinar$ with $|\pi_B(\alpha', \beta')|  \geq 2$.
Then by \cref{thm:blackbox}~\ref{thmitem:next_to_two}, for each of the other at least four pairs $\ab \in\Pinar$ there exists $v\in\mathcal{D}$ such that $\ab + v = \abprime$, since $\ab + v$ must satisfy $|\pi_B(\ab + v)| \geq 2$ and $\abprime$ is the only such pair.
But $\abprime$ is a vertex pair, and hence it can have at most three such adjacent pairs, a contradiction (also see \cref{fig:borderSquares1}).

\emph{Case 2: There are at most two vertex pairs with $|\bar{\pi}_B\ab| =1$.}
Consider a border pair $(\alpha', \beta')$ and $v\in\mathcal{D}$ such that $(\alpha', \beta')\pm v \in \Pinar$.
The tight constraint of $(\alpha', \beta')$ contains two vertex pairs.
Both of them must satisfy $|\bar{\pi}_B\ab|=1$, for otherwise, \cref{prop:averaging} would imply $|\bar{\pi}_B(\alpha', \beta')| \geq 2$, contradicting our assumption.
Additionally, since we assumed to have at most two vertex pairs with $|\bar{\pi}_B\ab| =1$, it follows that any vertex pair $\ab$ outside this tight constraint must satisfy $|\bar{\pi}_B\ab| \geq 2$.
But then, from an argument analogous to the one in the first case above, it follows that there can not be a border pair $\ab \in \Pinar$ that does not satisfy the tight constraint with equality, as otherwise we would get $|\bar{\pi}_B\ab| \geq 2$.
Consequently, all pairs $\ab \in \Pinar$ with $|\bar{\pi}_B \ab| = 1$ (those are the only ones where there might be a hidden solution) satisfy the constraint that is tight at $\abprime$ with equality.
Thus, we can let $\Pi'$ be the pattern shape defined by all $\ab$ that satisfy the tight constraint with equality (see \cref{fig:borderSquares2} for an example).
\end{proof}

\begin{figure}[ht]
\begin{subfigure}[t]{0.49\textwidth}
\begin{center}
\begin{tikzpicture}[scale=0.7]
\begin{scope}[every path/.append style={very thick}]
\draw (0,0) -- (3,0);
\draw (0,1) -- (3,1);
\draw (0,2) -- (2,2);
\draw (0,0) -- (0,2);
\draw (1,0) -- (1,2);
\draw (2,0) -- (2,2);
\draw (3,0) -- (3,1);
\draw[red] (0,0) -- (0,1) -- (1,1) -- (1,0) --cycle;

\node at (0.5, 0.5) {$1$};
\node at (1.5, 0.5) {$1$};
\node at (2.5, 0.5) {$1$};
\node at (0.5, 1.5) {$1$};
\node at (1.5, 1.5) {$2$};
\end{scope}

\begin{scope}[every path/.append style={gray}]
	\draw[->] (-0.5, -0.3) -- node[pos=1, below] {$\alpha$} (3.5,-0.3);
	\draw[->] (-0.3, -0.5) -- node[pos=1, left] {$\beta$} (-0.3,2.4);
	\foreach \tick/\num in {0.55/0,1.5/1,2.5/2} {
		\draw (\tick, -0.2) -- (\tick, -0.4) node[below, font=\scriptsize]{$\num$};
	}
	\draw (-0.2, 0.5) -- (-0.4, 0.5) node[left, font=\scriptsize]{$0$};
	\draw (-0.2, 1.5) -- (-0.4, 1.5) node[left, font=\scriptsize]{$1$};
\end{scope}

\end{tikzpicture}
\end{center}
\caption{The pair marked in red is not adjacent to any $\ab$ with $\pi_B\ab \geq 2$, hence this pattern cannot occur.}
\label{fig:borderSquares1}
\end{subfigure}
\begin{subfigure}[t]{0.49\textwidth}
\begin{center}
\begin{tikzpicture}[scale=0.7]
\begin{scope}[every path/.append style={very thick}]
\draw (0,0) -- (3,0);
\draw (0,1) -- (3,1);
\draw (0,2) -- (2,2);
\draw (0,0) -- (0,2);
\draw (1,0) -- (1,2);
\draw (2,0) -- (2,2);
\draw (3,0) -- (3,1);
\draw[red] (0,0) -- (0,1) -- (3,1) -- (3,0) --cycle;

\node at (0.5, 0.5) {$1$};
\node at (1.5, 0.5) {$1$};
\node at (2.5, 0.5) {$1$};
\node at (0.5, 1.5) {$2$};
\node at (1.5, 1.5) {$2$};
\end{scope}

\begin{scope}[every path/.append style={gray}]
	\draw[->] (-0.5, -0.3) -- node[pos=1, below] {$\alpha$} (3.5,-0.3);
	\draw[->] (-0.3, -0.5) -- node[pos=1, left] {$\beta$} (-0.3,2.4);
	\foreach \tick/\num in {0.55/0,1.5/1,2.5/2} {
		\draw (\tick, -0.2) -- (\tick, -0.4) node[below, font=\scriptsize]{$\num$};
	}
	\draw (-0.2, 0.5) -- (-0.4, 0.5) node[left, font=\scriptsize]{$0$};
	\draw (-0.2, 1.5) -- (-0.4, 1.5) node[left, font=\scriptsize]{$1$};
\end{scope}
\end{tikzpicture}
\end{center}
\caption{Hidden solutions may only occur for the pairs marked in red.}
\label{fig:borderSquares2}
\end{subfigure}
\caption{The two cases arising in the proof of \cref{lemma:easy_border}. Every square corresponds to a pair $\ab\in\Pinar$, and the numbers inside indicate the value of $|\bar{\pi}_B\ab|$.}
\label{fig:borderSquares}
\end{figure}

\subsubsection*{Pattern structure of \ref{item:type4}}

For \ref{item:type4} pattern structure, we first observe that, by \cref{obs:immediateConclusion}, if there are any solutions for pairs $\ab\in\Pinar$ with $|\pi_B\ab|\geq 2$, we can also find one efficiently by combining solutions computed for the $A$- and $B$-subproblem when determining $\bar\pi_A$ and $\bar\pi_B$.
The other case, i.e., when no solutions exist for such $\ab$, is covered by \cref{thm:blackbox}~\ref{thmitem:linear_pattern}. Again, that statement allows to reduce the problem to a new \GCTUF problem with the same group $G$ and at the same depth $d$, but at least one variable less. Therefore, it allows us to make progress with respect to the number of variables.

\subsection*{Summary}

The above discussion can be summarized in the following theorem.
Recall that the depth of an \RGCTUF problem is defined as $d\coloneqq|G|-|R|$.

\begin{theorem}\label{thm:summaryDecomp}
Let $G$ be a finite abelian group.
Consider an \RGCTUF{} problem $\mathcal{P}$ with $n$ variables, group $G$, set of target residues $R$, depth $d\leq 3$, and a constraint matrix $T$ that is a $3$-sum of two matrices with $n_A$ and $n_B$ many columns, respectively, such that $n=n_A+n_B$ and $n_A,n_B\geq 2$. %
Let $p \coloneqq \min \{n_A, n_B\}$.
Assume furthermore that there is no non-trivial subgroup $H$ of $G$ with $R = R + H$.
Then, after solving at most $(d+1)^3$ many \RGCTUF{} problems with $p$ variables, group $G$ and depth at most $d$, as well as at most $d(d+1)^2$ \RGCTUF{} problems with $n-p$ variables, group $G$ and depth at most $d-1$, one can either
\begin{itemize}
\item find a solution of $\mathcal{P}$ in strongly polynomial time, or
\item determine a single \RGCTUF{} problem $\mathcal{P}'$ with at most $n-p+1$ variables, group $G$ and depth $d$, such that $\mathcal{P}$ is feasible if and only if $\mathcal{P}'$ is feasible.
Additionally, a solution of $\mathcal{P}'$ can be transformed into a solution of $\mathcal{P}$ in strongly polynomial time.
\end{itemize}
Furthermore, all involved \RGCTUF{} problems can be constructed in strongly polynomial time.
\end{theorem}

\subsection{Proof of \cref{thm:reductionToBaseBlocksRGCTUF}}

Consider an \RGCTUF{} problem with group $G$, $n$ variables, and depth $d = |G|-|R|\leq 3$.
If $d=0$, then the problem is an unconstrained TU problem, and thus it is enough to find a vertex solution of the linear relaxation.
This can be done in strongly polynomial time using the algorithm of \textcite{tardos_1986_strongly}.
If $d>0$, we apply \cref{thm:TUdecomp} to the constraint matrix $T$.
If $T$ is covered by one of \cref{thmitem:TUdecomp_netw,thmitem:TUdecomp_const}, then it is a base block matrix itself, so a single call to the oracle suffices.
Else, $T$ is covered by \cref{thmitem:TUdecomp_pivotsum} of \cref{thm:TUdecomp}, and we may assume that $T$ decomposes into a $3$-sum of two matrices with at least two columns each (see \cref{sec:pivot} for the case where a pivot step is necessary).
In particular, note that in this case, we have $n\geq 4$.
Using \cref{lem:cosetReduction} or \cref{thm:summaryDecomp}, we now reduce the problem to one or more smaller problems, until we eventually obtain base block problems, which we solve by an oracle call.

We bound the number of oracle calls triggered by our procedure.
Let $f(n,d)$ be the smallest upper bound on the number of oracle calls when starting from an instance with $n$ variables and depth $d$.
We claim that
\[
	f(n,d) \leq \left(d+1\right)^{3d}n^{d+3\log_2 (d+1) + 2}\enspace.
\]
To prove this bound, we use induction on $n+d$.
First observe that, by the above discussion, $f(n,0)=0$ for any $n\geq 1$, and $f(n,d)=1$ for $n\leq 3$ and $d>0$ (in the latter case, we cannot attain \cref{thmitem:TUdecomp_pivotsum} of \cref{thm:TUdecomp}).
Now consider an \RGCTUF{} problem with $n$ variables and depth $d$.
If $R=R+H$ for a non-trivial subgroup $H$ of $G$, \cref{lem:cosetReduction} allows for a reduction to a single \RGCTUF problem at smaller depth.
Thus, we end up with at most $f(n,d')$ many base block problems for some $d'<d$.
In this case, the induction hypothesis implies the claimed bound on $f(n,d)$ because it is monotone in $d$.
In the other case, we apply \cref{thm:summaryDecomp}.
Thus, there is a $p \in \{2, \dots, \left\lfloor \sfrac{n}{2}\right\rfloor\}$ such that the number of base block problems we have to solve is bounded by
\begin{multline*}
	(d+1)^3 f(p,d) + d(d+1)^2 f(n-p,d-1) + f(n-p+1,d)\\
			\leq \left(d+1\right)^{3d}n^{d+3\log_2 (d+1) + 2}
			\underbrace{\left( \left(\frac{p}{n}\right)^2 	+ \frac{n-p}{n^2} + \left(\frac{n-p+1}{n}\right)^2 \right)}_{\leq 1}
			\leq \left(d+1\right)^{3d}n^{d+3\log_2 (d+1) + 2}\enspace,
\end{multline*}
proving the claimed bound on $f(n,d)$.
Here, we use that
\begin{align*}
	(d+1)^3 f(p,d) &\leq \left(d+1\right)^{3d+3}p^{d+3\log_2 (d+1) + 2}\\
	& = \underbrace{(d+1)^3 \left(\frac{p}{n}\right)^{d+3\log_2 (d+1)}}_{\leq 1}\left(d+1\right)^{3d}n^{d+3\log_2 (d+1) + 2} \left(\frac{p}{n}\right)^{2} \\
	&\leq \left(d+1\right)^{3d}n^{d+3\log_2 (d+1) + 2}\left(\frac{p}{n}\right)^2\enspace,
\end{align*}
as $\sfrac{p}{n}\leq \sfrac{1}{2}$ and hence $\left(\sfrac{p}{n}\right)^{d+3\log_2 (d+1)} \leq (d+1)^{-3}$, together with
\begin{align}
\nonumber
	d(d+1)^2 f(n-p,d-1) &\leq d(d+1)^2 d^{3d-3}(n-p)^{d-1+3\log_2 d + 2}\\
\nonumber
	&\leq \underbrace{\frac{d(d+1)^2}{(d+1)^3} \left(\frac{n-p}{n}\right)^{d+3\log_2 (d+1)}}_{\leq 1}\left(d+1\right)^{3d}n^{d+3\log_2 (d+1) + 2} \left(\frac{n-p}{n^{2}}\right)\\
\tag*{\qed}
    &\leq \left(d+1\right)^{3d}n^{d+3\log_2 (d+1) + 2} \left(\frac{n-p}{n^{2}}\right)\enspace.
\end{align}

\begingroup
\setlength{\emergencystretch}{0.5em}
\printbibliography
\endgroup

\appendix
\section{Adapted proofs of structural results}

This section is devoted to the proofs of \cref{lemma:boundedPatternShape,thm:blackbox}.
More precisely, given that these statements were proved in a congruency-constrained setting in~\cite{nagele_2022_congruency}, we mainly recall proof ideas and comment on how to adapt them to the group constraint setting that we work with.

\subsection{Proof of \cref{lemma:boundedPatternShape}}\label{sec:pf_boundedPattern}

\cref{lemma:boundedPatternShape} is the group-constrained analogue of~\cite[Lemmas 2.5 and 5.1]{nagele_2022_congruency}.
To obtain the result, \cite{nagele_2022_congruency} exploits a decomposition theorem for solutions of totally unimodular systems~\cite[Lemma 2.1]{nagele_2022_congruency} combined with the following fact for the special case of $G=\inlinequot{\mathbb{Z}}{m\mathbb{Z}}$~\cite[Lemma 2.2]{nagele_2022_congruency}, which is indeed the only property of $\inlinequot{\mathbb{Z}}{m\mathbb{Z}}$ used throughout the proof.

\begin{lemma}\label{lem:sumModM}
Let $G$ be a finite abelian group, $R\subseteq G$, and $r_1,\ldots,r_\ell\in G$ with $\sum_{i\in[\ell]} r_i \in R$.
If there is no interval $I=\{i_1,\ldots,i_2\}$ with $i_1,i_2\in[\ell]$ and $i_1<i_2$ such that $\sum_{i\in[\ell]\setminus I}r_i\in R$, then $\ell\leq|G|-|R|$.
\end{lemma}

Consequently, through \cref{lem:sumModM}, the proofs of~\cite[Lemmas 2.5 and 5.1]{nagele_2022_congruency} immediately extend to the group constraint setting, and thereby imply \cref{lemma:boundedPatternShape}.
We also remark that \cref{lem:sumModM} is in fact just a slightly generalized and more constructive version of \cref{lem:elementary}.
Moreover, the original proof for $G=\inlinequot{\mathbb{Z}}{m\mathbb{Z}}$ directly generalizes.
We repeat it here for completeness.

\begin{proof}[Proof of \cref{lem:sumModM}]
Assume for the sake of deriving a contradiction that there is no interval $I\subseteq [\ell]$ with $\sum_{i\in[\ell]\setminus I}r_i\in R$, but $\ell \geq |G|-|R|+1$.
Consider the $\ell$ group elements $s_0=0$, $s_1=r_1$, \ldots, $s_{\ell-1}=r_1+\ldots+r_{\ell-1}$.
Observe that $s_j\notin R$ for all $j\in[\ell-1]$; for otherwise, there is an interval $I=\{j+1,\ldots,\ell\}$ for some $j\in[\ell-1]$ such that $\sum_{i\in[\ell]\setminus I} = s_j\in R$, contradicting the assumption.
Thus, $s_j\in G\setminus R$ for all $j\in[\ell-1]$.
Hence, because $\ell\geq |G|-|R|+1$, we have by the pigeonhole principle that there exist $j_1,j_2\in[\ell-1]$ with $j_1<j_2$ such that $s_{j_1}=s_{j_2}$.
Thus, $I=\{j_1+1, \ldots, j_2\}$ is an interval with $\sum_{i\in[\ell]\setminus I} = \sum_{i\in[\ell] I}r_i - (s_{j_2}-s_{j_1}) = \sum_{i\in[\ell]\setminus I} r_i \in R$, again contradicting the assumption and hence completing the proof.
\end{proof}

\subsection{Proof of \cref{thm:blackbox}}\label{sec:pf_blackbox}
\newcommand{\abn}[1]{(\alpha_{#1}, \beta_{#1})}

The statements in \cref{thm:blackbox} are closely linked to a \GCTUF problem in the form presented in~\eqref{eq:structured-problem}.
In the following, we say that a solution $x=(x_A, x_B)$ of a problem in that form is \emph{a solution for $\ab$} if $\alpha = f^\top x_B$ and $\beta = h^\top x_A$.

\subsubsection*{\cref{thmitem:next_to_two} of \cref{thm:blackbox}}

In the congruency-constrained setting, the first part of \cref{thm:blackbox}~\ref{thmitem:next_to_two} was proved in~\cite[Lemma 5.9]{nagele_2022_congruency}, while the second part is implicit in~\cite[Proof of Lemma 5.8]{nagele_2022_congruency}.
The argument uses a generalized version of \cref{prop:averaging}, which states that given distinct pairs $\abn{1}, \abn{2}\in \Pinar$, there are $\abn{3}, \abn{4}\in\Pinar$ such that for any solutions $x_1$ for $\abn{1}$ and $x_2$ for $\abn{2}$, there exist solutions $x_3$ for $\abn{3}$ and $x_4$ for $\abn{4}$ such that $x_1+x_2= x_3+x_4$ (in fact, $\abn{3}$ and $\abn{4}$ are equal to $\frac12(\abn{1}+\abn{2})$ up to rounding, see~\cite[Lemma 5.3]{nagele_2022_congruency}).

Towards a proof of \cref{thm:blackbox}~\ref{thmitem:next_to_two}, we may assume that $\abprime\neq \ab + v$ for all $v\in\mathcal{D}\cup \{0\}$, and apply the above with $\abn{1}=\abprime$ and $\abn{2}=\ab$.
This gives that $\abn{3}$ and $\abn{4}$ are both different from $\abn{1}$ and $\abn{2}$, but ``closer'' to $\abprime$ than $\ab$ was.
The result will follow by iteratively applying this argument after showing that $|\pi_B\abn{3}|\geq 2$ or $|\pi_B\abn{4}|\geq 2$.
In the congruency-constrained case, the latter is concluded from the assumption that $|\pi_B \ab|\geq 2$, i.e., that there are two feasible residues at $\ab$.
Indeed, if we had $|\pi_B\abn{3}|=|\pi_B\abn{4}|=1$, then $\gamma^\top(x_3+x_4)$---and hence also $\gamma^\top(x_1+x_2)$---would have to yield the same residue for all solutions $x_1$ for $\abn{1}$ and $x_2$ for $\abn{2}$.
This reasoning holds analogously for any finite abelian group other than the cyclic groups $\inlinequot{\mathbb{Z}}{m\mathbb{Z}}$, by well-definedness of addition to be precise.
Because no other properties of congruency-constraints are exploited, the proofs of~\cite{nagele_2022_congruency} directly translate to a proof of \cref{thm:blackbox}~\ref{thmitem:next_to_two}.

\subsubsection*{\cref{thmitem:linear_pattern} of \cref{thm:blackbox}}\label{sec:blackboxTypeIandIV}

\cref{thm:blackbox}~\ref{thmitem:linear_pattern} gives two sufficient conditions that allow for reduction to a problem with fewer variables.
The conditions are precisely that the pattern structure is of \ref{item:type1}, or that it is of \ref{item:type4} and that there are no solutions for pairs $\ab\in\Pinar$ with $|\pi\ab|\geq 2$.

For pattern structure of \ref{item:type1}, \cite[Corollary 5.1]{nagele_2022_congruency} shows that---in the congruency-constrained setting---there exist $r_0, r_1, r_2\in\inlinequot{\mathbb{Z}}{m\mathbb{Z}}$ such that $\pi_B\ab = \{r_0 + r_1\alpha + r_2\beta\}$ for all $\ab \in \Pinar$, and $\pi_B$ is called \emph{linear} in this case.
Again, the proof is based on an averaging argument (i.e., \cref{prop:averaging}) that seamlessly carries over to the more general group-constrained setting with a finite abelian group $G$ by simply replacing calculations in $\inlinequot{\mathbb{Z}}{m\mathbb{Z}}$ by calculations in $G$.
The same applies to showing that in case of a linear pattern, a \GCTUF problem can be reduced to an equivalent problem with the same group $G$, at the same depth $d$, and strictly fewer variables~\cite[Theorems~2.4 and~2.5]{nagele_2022_congruency}.

For pattern structure of \ref{item:type4}, \cref{claim:4vertexSquares} gives that $\Pinar$ consists of at most four pairs, and all of them are vertex pairs.
This pattern family is very restricted, and is in fact a subset of the patterns covered by an analysis of certain small patterns in~\cite[Proof of Lemma 5.11]{nagele_2022_congruency}.
Concretely, for the \ref{item:type4} pattern structure that we consider here, it was concluded (using another averaging argument, i.e., \cref{prop:averaging}) that one can choose for every $\ab\in\Pinar$ a singleton-subset $\tilde\pi_B\ab\subseteq\pi_B\ab$ such that $\tilde\pi_B$ is linear in the sense introduced above.
Now recall that we also assume here that there are no solutions for pairs $\ab\in\Pinar$ with $|\pi\ab|\geq 2$, i.e., solutions can only occur for $\ab\in\Pinar$ with $|\pi_B\ab|=1$.
For those $\ab$, we have $\pi_B\ab=\tilde\pi_B\ab$, so it is enough to look for solutions compatible with $\tilde\pi_B$.
But then, linearity of $\tilde\pi_B$ allows for the same reduction to an equivalent problem with fewer variables as discussed above (the congruency-constrained version is given in~\cite[Theorems~5.2 and~2.5]{nagele_2022_congruency}).
Again, in all involved proofs of~\cite{nagele_2022_congruency}, all calculations in $\inlinequot{\mathbb{Z}}{m\mathbb{Z}}$ can directly be replaced by calculations in any fixed finite abelian group $G$ without affecting correctness of the proofs, hence the results carry over as desired.\qed
\section{Adapted proofs of base block results}\label{sec:BBreduction}

In this section, we discuss how to extend the proofs of~\cite[Section 4]{nagele_2022_congruency} for congruency-constrained TU problems with a base block constraint matrix, i.e., matrices covered by \cref{thmitem:TUdecomp_netw,thmitem:TUdecomp_const} of \cref{thm:TUdecomp}.
The congruency constraints that were used previously may equivalently be formulated as constraints in a cyclic group $\inlinequot{\mathbb{Z}}{m\mathbb{Z}}$, and it turns out that all arguments extend to general finite abelian groups, i.e., the setting that we need for our purposes.
In this appendix, we recall the proofs of~\cite{nagele_2022_congruency} and comment on the mostly straightforward modifications for the sake of completeness.

To this end, let us first define \emph{group-constrained TU problems} (\linkdest{prb:GCTU}\GCTU problems) to be the optimization variant of \GCTUF{} problems, i.e., where additionally to the \GCTUF setup, we are given an objective $c\in\mathbb{R}^n$ that we want to minimize over all feasible solutions of the \GCTUF problem.
Note that we can always assume to start with a \GCTU{} problem whose relaxation (i.e., the problem obtained after dropping the group constraint) is feasible, which we can check in strongly polynomial time; for otherwise, the \GCTU{} problem is clearly infeasible.
Hence, we assume feasibility of the relaxation throughout this section.
To start with, let us recall the definition of a \emph{network matrix}.
\begin{definition}\label{def:networkMatrix}
	A matrix $T$ is a \emph{network matrix} if the rows of $T$ can be indexed by the edges of a directed spanning tree $(V, U)$, and the columns can be indexed by the edges of a directed graph $(V, A)$ on the same vertex set, such that for every arc $a=(v, w)\in A$ and every arc $u\in U$,
	$$
	T_{u,a} = \begin{cases}
		1 & \text{if the unique $v$-$w$ path in $U$ passes through $u$ forwardly,}\\
		0 & \text{if the unique $v$-$w$ path in $U$ does not pass through $u$,}\\
		-1 & \text{if the unique $v$-$w$ path in $U$ passes through $u$ backwardly.}
	\end{cases}
	$$
\end{definition}

In the subsequent three sections, we distinguish the three different cases of base block matrices, namely whether the constraint matrix $T$ of the \GCTU{} problem that we consider is a network matrix, the transpose of a network matrix, or a matrix stemming from the constant-size matrices given in \cref{thmitem:TUdecomp_const} of \cref{thm:TUdecomp}.
\subsection{Network matrix base block GCTU problems}

In this section, we discuss the extension from congruency constraints to group constraints in the case where the constraint matrix is a network matrix.
Concretely, this will provide the following result.

\begin{theorem}\label{thm:solveNetwRPP}
Let $G$ be a finite abelian group.
There is a strongly polynomial time randomized algorithm for \GCTU{} problems with group $G$, unary encoded objectives, and constraint matrices that are network matrices.
\end{theorem}

In this case, one can exploit the graph structure that comes with network matrices to interpret \GCTU{} problems with network constraint matrices as minimum-cost group-constrained circulation problems in certain directed graphs. To get started, let us recall that a circulation $f$ in a directed graph $H=(V,A)$ with capacities $u\colon A\to\mathbb{Z}_{\geq 0}$ is a mapping $f\colon A\to\mathbb{Z}_{\geq 0}$ such that $f(a)\leq u(a)$ for every arc $a\in A$, and $f(\delta^+(v)) = f(\delta^-(v))$ for every vertex $v\in V$. Given arc lengths $\ell \colon A\to\mathbb{Z}$, the length of a circulation $f$ is $\ell(f)\coloneqq\sum_{a\in A}\ell(a)f(a)$. Note that here, arc lengths are allowed to be negative.
A group-constrained circulation problem is formally defined as follows.
\begin{mdframed}[innerleftmargin=0.5em, innertopmargin=0.5em, innerrightmargin=0.5em, innerbottommargin=0.5em, userdefinedwidth=0.95\linewidth, align=center]
	{\textbf{Group-Constrained Circulation (\linkdest{prb:GCC}\GCC):}}
	Let $H=(V,A)$ be a directed graph with capacities $u\colon A\to\mathbb{Z}_{\geq 0}$, arc lengths $\ell\colon A\to\mathbb{Z}$, and let $G$ be a finite abelian group, $\eta\colon A\to G$ and $r\in G$. Find a minimum-length circulation $f\colon A\to\mathbb{Z}_{\geq 0}$ in the given network such that $\sum_{a\in A}\eta(a)f(a) = r$.
\end{mdframed}

\noindent In~\cite[Lemma 4.2]{nagele_2022_congruency}, a reduction from \GCTU{} problems with a network matrix as constraint matrix to \GCC{} is presented for the case of a cyclic group $G=\inlinequot{\mathbb{Z}}{m\mathbb{Z}}$.
In this reduction, the only property of $\inlinequot{\mathbb{Z}}{m\mathbb{Z}}$ that is exploited is~\cite[Lemma 2.2]{nagele_2022_congruency}.
Thus, replacing this statement by its group-constrained analogue \cref{lem:sumModM} proved earlier immediately gives the following.
\begin{lemma}\label{lem:GCTUtoCCC}
	\GCTU{} problems with group $G$, objective vector $c$, and constraint matrices that are network matrices can be reduced in strongly polynomial time to \GCC{} problems with group $G$, capacities $u$ within $\{0, \ldots, |G|-1\}$, and arc lengths $\ell$ with $\|\ell\|_\infty\leq \|c\|_{\infty}$.
\end{lemma}
Now, one can finish the proof of \cref{thm:solveNetwRPP} by exploiting a connection to the so-called \emph{exact length circulation problem}, where the goal is to find a circulation whose length is equal to a given value.
The reduction follows the one given in~\cite[Lemma 4.3]{nagele_2022_congruency}, but needs to be slightly adapted in order to capture group constraints.

\begin{mdframed}[innerleftmargin=0.5em, innertopmargin=0.5em, innerrightmargin=0.5em, innerbottommargin=0.5em, userdefinedwidth=0.95\linewidth, align=center]
	{\textbf{Exact Length Circulation (\linkdest{prb:XLC}\XLC):}}
	Let $H=(V,A)$ be a digraph with capacities $u\colon A\to\mathbb{Z}_{>0}$ and arc lengths $\ell \colon A\to\mathbb{Z}$. Given $L\in\mathbb{Z}$, find a circulation $f$ in the given network such that $\ell(f)=L$.
\end{mdframed}
Exact length circulation problems can be solved using a randomized pseudopolynomial algorithm, as shown by \textcite{camerini_1992_rpp}.
They reduce the problem to an exact cost perfect matching problem, which can then be reduced to computing the coefficients of a well-defined polynomial.
The following theorem summarizes the result of \textcite{camerini_1992_rpp} for \XLC{}.
\begin{theorem}[\cite{camerini_1992_rpp}]\label{thm:XLCalgo}
	There is a randomized algorithm for \XLC{} problems in a directed graph $H=(V,A)$ with capacities $u\colon A\to\mathbb{Z}_{\geq 0}$ in time $\mathrm{poly}(|V|, \max_{a\in A}u(a), \max_{a\in A}|\ell(a)|)$.
\end{theorem}

Thus, it remains to extend the connection between \GCC{} and \XLC{} problems from the congruency-constrained to the group-constrained setting.
To this end, we recall the reduction for a single congruency constraint modulo $m$: Taking an integer representative in $\{0, \dots, m-1\}$ for each residue modulo $m$, we may construct new arc lengths that capture the residue representative $r(a)$ and original arc lengths $\ell(a)$ independently (e.g., by using $\tilde\ell(a) = c\cdot \ell(a) + r(a)$ for a large enough constant $c$).
Guessing both the length and the residue class of an optimal solution of the \GCC problem is then equivalent to guessing a target $L$ with respect to the new lengths, hence the reduction can be achieved by binary search.

For group constraints, we recall that we may equivalently interpret a group constraint as multiple congruency constraints, which can then be integrated into new arc lengths $\tilde\ell$ at different orders of magnitude.
Formally, this leads to the following lemma, which is an analogue of~\cite[Lemma 4.3]{nagele_2022_congruency}.
\begin{lemma}\label{lem:CCCtoXLC}
A \GCC{} problem in a graph $H=(V,A)$ with group $G$, arc lengths $\ell\colon A\to\mathbb{Z}$, and capacities $u\colon A\to\{0, 1, \ldots, |H|-1\}$ can be polynomially reduced to $\mathrm{poly}(|G|, |V|, |A|, \max_{a\in A}|\ell(a)|)$ many \XLC{} problems in $G$ with the same capacities.
\end{lemma}

\begin{proof}
First, we may decompose $G\cong\inlinequot{\Z}{m_1\Z} \times \cdots \times \inlinequot{\Z}{m_k\Z}$ for some $k\in\mathbb{Z}_{\geq 1}$.
For every $i\in[k]$, let $\varphi_i \colon G \rightarrow \inlinequot{\Z}{m_i\Z} \rightarrow \Z$ denote the corresponding natural projection maps composed with the natural mapping of a residue in $\inlinequot{\Z}{m_i\Z}$ to a representative in $\{0, \dots, m_i-1\}$.
Set $m\coloneqq|G|$.

Define $\tilde\ell\colon A\to\mathbb{Z}$ for every arc $a\in A$ as
$$\tilde{\ell}(a)=\ell(a)\cdot m^{2k}|A|^k + \sum_{i=1}^k m^{2i-2}|A|^{i-1}\varphi_i(\eta(a))\enspace.$$
We thus have $\tilde{\ell}(f) = \ell(f)\cdot m^{2k}|A|^k + \sum_{i=1}^k m^{2i-2}|A|^{i-1}\varphi_i(\eta(f))$.
Observe that $\sum_{a\in A}\varphi_i(\eta(a))f(a) < m^2 |A|$ for every $1\leq i\leq k$, hence  from $\tilde{\ell}(f)$, one can retrieve $\ell(f)$ as well as $\sum_{a\in A}\varphi_i(\eta(a))f(a)$ for every $i\in[k]$.
Consequently, finding a circulation of length $L$ with $\sum_{a\in A}\eta(a)f(a)= r$ is equivalent to solving \XLC{} problems in $H$ with respect to lengths $\tilde{\ell}$ and with target length $\tilde{L}=L\cdot m^{2k}|A|^k + \sum_{i=1}^km^{2i-2}|A|^{i-1} (d_im_i + \varphi_i(r))$ for all tuples $(d_1, \dots, d_k)\in\{0,\ldots, m|A|-1\}^k$.
We can find the smallest $L$ for which there is a \GCC{} solution of length $L$ by binary search in $O(k\log (m|A|\cdot\max_{a\in A}|\ell(a)|))$ iterations, because $|\ell(f)| = \left|\sum_{a\in A}\ell(a)f(a)\right| \leq m|A|\cdot\max_{a\in A}|\ell(a)|$.
Altogether, this gives the desired result.
\end{proof}

Combining \cref{lem:GCTUtoCCC,thm:XLCalgo,lem:CCCtoXLC} readily implies \cref{thm:solveNetwRPP}.
\subsection{Transposed network matrix base block GCTU problems}\label{sec:transposedNetworkBB}

For \GCTU problems with a constraint matrix that is the transpose of a network matrix, we recall from \cref{sec:transposedNetworkBaseBlock} that we aim for a reduction to a lattice problem of the form given in~\eqref{eq:CCSM}.
For the sake of completeness, let us restate that problem in the optimization setting here.

\begin{mdframed}[innerleftmargin=0.5em, innertopmargin=0.5em, innerrightmargin=0.5em, innerbottommargin=0.5em, userdefinedwidth=0.95\linewidth, align=center]
	{\textbf{Group-Constrained Lattice Optimization (\linkdest{prb:GCLO}{\GCLO}):}}
	Let $N$ be a finite set, $\mathcal{L}\subseteq 2^N$ a lattice, $(G,+)$ a finite abelian group, $\gamma \colon N \rightarrow G$, $r\in G$, and $w\colon N \rightarrow \mathbb{Z}$.
	The task is to find $X\in\mathcal{L}$ with $\gamma(X)\coloneqq\sum_{x\in X}\gamma(x) = r$ minimizing $\sum_{x\in X} w(x)$, or decide infeasibility.
\end{mdframed}

\noindent
In~\cite[Section 4.2]{nagele_2022_congruency}, the following result was shown implicitly for the case of a cyclic group $G$.
It turns out, though, that no properties of cyclic groups beyond them being finite abelian groups are used in the proofs, hence by substituting all occurrences of congruency constraints, i.e., constraints in a cyclic group $\inlinequot{\mathbb{Z}}{m\mathbb{Z}}$, with group constraints, the original proofs also imply the theorem for a general finite abelian group.

\begin{theorem}\label{thm:transposeNetworkToLattice}
  Given a finite abelian group $G$, consider a \GCTU problem on $n$ variables and a constraint matrix that is the transpose of a network matrix.
  One can in strongly polynomial time determine a \GCLO problem over a ground set $N$ with $|N| = n|G|$ such that from an optimal solution of the \GCLO problem, we can in strongly polynomial time compute an optimal solution of the \GCTU problem.
\end{theorem}

In particular, \cref{thm:transposeNetworkToLattice} implies the reduction claimed in \cref{prop:GCTUFtoGCLF} by observing that we can in strongly polynomial time determine $w\colon N \rightarrow \mathbb R$ such that $f(A) = \sum_{x\in A} w(x)$. %
\subsection{Matrices stemming from particular constant-size matrices}\label{sec:constCore}

To complete the discussion of base block \GCTU{} problems, we now cover \GCTU{} problems with constraint matrices covered by \cref{thmitem:TUdecomp_const} of \cref{thm:TUdecomp}.
In other words, these are matrices $T$ that can be obtained from the two matrices
\begin{equation}\label{eq:specialMatrices}
	\begin{pmatrix*}[r]
		1 & -1 &  0 &  0 & -1 \\
		-1 &  1 & -1 &  0 &  0 \\
		0 & -1 &  1 & -1 &  0 \\
		0 &  0 & -1 &  1 & -1 \\
		-1 &  0 &  0 & -1 &  1
	\end{pmatrix*}
	\quad\text{and}\quad
	\begin{pmatrix}
		1 &  1 &  1 &  1 &  1 \\
		1 &  1 &  1 &  0 &  0 \\
		1 &  0 &  1 &  1 &  0 \\
		1 &  0 &  0 &  1 &  1 \\
		1 &  1 &  0 &  0 &  1
	\end{pmatrix}
\end{equation}
by repeatedly appending unit vector rows or columns, appending a copy of a row or column, and inverting the sign of a row or column.

To tackle problems with such constraint matrices, \cite[Section 4.3]{nagele_2022_congruency} more generally considers constraint matrices $T$ that are obtained from a \emph{core} matrix $C$ through the above operations.
(Concretely, in the case relevant here, $C$ will be one of the two matrices in~\eqref{eq:specialMatrices}.)
It is shown that if $C$ has $\ell$ columns, then there are integral vectors  $s_1,\ldots,s_\ell$ such that once the values $s_i^\top x$ are known for all $i\in[\ell]$, the inequality system $Tx\leq b$ can be reduced to an equivalent system $T'x\leq b'$ on the same variables $x$, where $T'$ is a network matrix and the transpose of a network matrix at the same time, and $b'$ is integral.
By exploiting a proximity statement for congruency-constrained TU problems, \cite{nagele_2022_congruency} observe that it suffices to consider values $s_i^\top x\in\{-m+1,\ldots,m-1\}$, where $m$ is the modulus of the congruency-constraint, and they conclude that there are only $(2m-1)^\ell$ possible combinations of values for $s_i^\top x$ to be enumerated, which can be done in polynomial time for constant $\ell$.
To extend this reasoning from congruency constraints (i.e., constraints in a cyclic group $G=\inlinequot{\mathbb{Z}}{m\mathbb{Z}}$) to constraints in general finite abelian groups, we observe that the aforementioned proximity statement relies on a decomposition theorem for solutions of totally unimodular systems~\cite[Lemma 2.1]{nagele_2022_congruency} and the property of cyclic groups that we generalized in \cref{lem:sumModM}.
Thereby, it is again a matter of replacing modular arithmetic by calculations in a general finite abelian group $G$ to obtain the following result.

\begin{lemma}\label{lem:reductionConstCore}
Let $G$ be a finite abelian group and consider a \GCTU{} problem with group $G$ and a constraint matrix $T$ that can be obtained from a matrix $C$ with $\ell$ columns by repeatedly appending unit vector rows or columns, appending a copy of a row or column, and inverting the sign of a row or column.
Then, the \GCTU{} problem can be reduced to $(2|G|-1)^\ell$ many \GCTU{} problems with group $G$ and constraint matrices of size linear in the size of $T$ that are network matrices and transposes of network matrices at the same time.
\end{lemma}

For feasibility problems, we may thus exploit \cref{thm:transpose-bb} to obtain the following direct corollary.

\begin{corollary}
Let $G$ be a finite abelian group. There is a strongly polynomial time algorithm for solving \GCTUF problems with group $G$ and a constraint matrix covered by \cref{thmitem:TUdecomp_const} of \cref{thm:TUdecomp}.
\end{corollary}
\section{Pivoting steps}\label{sec:pivot}

By \cref{thm:TUdecomp} (concretely, \cref{thmitem:TUdecomp_pivotsum}), we may face pivoting operations when applying Seymour's decomposition.
For congruency-constrained TU problems, \cite[Theorem 2.7]{nagele_2022_congruency} shows that such operations can be dealt with in the following sense:
After the addition of a single variable upper bound to a given congruency-constrained TU problem, there is a unimodular variable transformation that transforms the problem into an equivalent congruency-constrained TU problem such that (up to one extra constraint that is an upper bound on a variable), the new constraint matrix is the desired pivoted form of the original constraint matrix.

The transformation argument sketched above only exploits that congruency constraints are constraints on a linear combination of the variables, and thus immediately extends to group constraints.
Adding an upper bound constraint on a variable can be done without changing the problem due to a proximity result for congruency-constrained TU problems~\cite[Lemma 3.2]{nagele_2022_congruency}.
The latter is again based on a decomposition theorem for solutions of TU systems~\cite[Lemma 2.1]{nagele_2022_congruency} and the property of cyclic groups that we generalized in \cref{lem:sumModM}, and thus translates to group constraints.

\end{document}